\documentclass[ba,preprint]{imsart}

\RequirePackage{amsthm,amsmath,amsfonts,amssymb}
\RequirePackage[authoryear]{natbib}
\RequirePackage[colorlinks,citecolor=blue,urlcolor=blue,backref=page,backref=page]{hyperref}
\RequirePackage{graphicx}
\RequirePackage{tikz}
\usepackage{tabularx}
\usepackage{algorithm}
\usepackage{algpseudocode}
\usepackage{booktabs}
\usepackage{float}
\usepackage{pgfplots}
\pgfplotsset{compat=1.18}

\let\oldappendix\appendix 
\renewcommand{\appendix}{%
  \oldappendix
  \renewcommand{\thesubsection}{Appendix \Alph{subsection}}
}

\firstpage{1}

\lastpage{1}

\startlocaldefs
\theoremstyle{plain}

\newtheorem{theorem}{Theorem}
\newtheorem{proposition}{Proposition}
\newtheorem{lemma}{Lemma}
\newtheorem{corollary}{Corollary}
\theoremstyle{definition}
\newtheorem{definition}{Definition}
\newtheorem{example}{Example}
\newtheorem{property}{Property}

\theoremstyle{remark}

\endlocaldefs

\begin{document}

\begin{frontmatter}
\title{Context Tree Prior Distributions based on Node Weighting with exact Bayes Factors}
\runtitle{Context Tree Priors based on Node Weighting}

\begin{aug}
\author[A]{\fnms{Thiago}~\snm{Paulichen}\ead[label=e1]{thiagopaulichen@gmail.com}} \and
\author[B]{\fnms{Victor}~\snm{Freguglia}\ead[label=e2]{victorfs@unicamp.br}}
\address[A]{Department of Statistics,
State University of Campinas\printead[presep={,\ }]{e1}}

\address[B]{Department of Statistics,
State University of Campinas\printead[presep={,\ }]{e2}}
\runauthor{T. Paulichen and V. Freguglia}
\end{aug}

\begin{abstract}
Variable-length Markov chains (VLMCs) are a flexible class of higher-order Markov models that admit a natural representation as context trees. Existing Bayesian methods for specifying prior distributions on trees rely on branching processes, but these suffer from a fundamental limitation: the connection between node-branching probabilities and the structural properties of the induced tree distribution is not straightforward, making it difficult to encode specific structural beliefs. We address this issue by introducing a novel representation of prior distributions on tree spaces, characterized by assigning weights to individual contexts through a function on nodes. In this way, our approach provides an intuitive mechanism for incorporating structural hypotheses into the prior while preserving computational tractability, allowing marginal likelihoods and posterior mode trees to be computed exactly via generalizations of the Context Tree Weighting (CTW) and Context Tree Maximizing (CTM) algorithms. By enabling exact Bayes factor calculations, our methodology provides a principled framework for model comparison over structural priors. We demonstrate the flexibility and effectiveness of our approach by comparing different prior specifications through simulation studies and an application to financial markets.
\end{abstract}

\begin{keyword}
\kwd{Bayesian context trees}
\kwd{Bayes factors}
\kwd{Context-tree distributions}
\kwd{Context-tree weighting}
\kwd{Model selection}
\end{keyword}

\end{frontmatter}

\section{Introduction}

Variable-Length Markov Chains (VLMCs), originally introduced by \citet{Riss}, provide a flexible framework for modeling discrete stochastic processes whose memory length varies according to the observed past.
Instead of assuming a fixed-order Markov dependence, VLMCs represent the dependence structure through context trees, yielding parsimonious models that automatically adapt their complexity to the underlying process.
Since their introduction, VLMCs have found successful applications in data compression, linguistics, neuroscience, and financial time series, while their statistical properties have been extensively investigated in the literature \citep{Buhl, collet2008random}.

Bayesian inference for context trees has benefited from a sequence of developments exploiting the recursive structure of the model.
The Context Tree Weighting (CTW) algorithm \citep{CTW} demonstrated that recursive computations can efficiently aggregate information over the exponentially large space of context trees, although it was originally proposed from a predictive rather than Bayesian perspective. Building on these ideas, Bayesian Context Trees (BCT) \citep{kontoyiannis2022bayesiancontexttreesmodelling} introduced an explicit Bayesian formulation, enabling exact computation of posterior quantities under suitable prior distributions. More recently, distributions on rooted trees based on branching probabilities \citep{Nakahara_2022} have considerably expanded the class of prior models available for tree structures.

These developments have established a powerful computational framework for Baye-sian inference over context trees.
However, they leave open a different methodological question that has received comparatively less attention: how should prior distributions over context trees be specified? Existing Bayesian formulations typically construct prior distributions through recursive branching probabilities.
Although this construction leads to elegant recursive algorithms, it provides only indirect control over the structural properties of the induced distribution over trees.
As a consequence, prior beliefs that arise naturally in applications, such as preferring shallow trees, encouraging sparse structures, concentrating probability around a particular depth or prohibiting certain structures, are difficult to express directly.

This limitation is not merely a matter of convenience, as prior distributions should reflect interpretable prior beliefs about the object of interest. For context-tree models, these beliefs are naturally formulated in terms of structural properties of the tree itself rather than local branching decisions. The current parametrizations therefore create a disconnect between prior elicitation and prior representation: the analyst thinks in terms of tree structures but is required to specify branching probabilities whose global consequences are often difficult to anticipate.

The goal of this work is to bridge this gap by introducing an alternative representation for prior distributions on context trees.
Instead of parametrizing priors through branching probabilities, we assign non-negative weights directly to the nodes of the tree, defining a class of \emph{context-tree functions} that naturally induces probability distributions over admissible context trees.
This representation provides a flexible language for constructing structurally meaningful priors while preserving the recursive structure required for efficient Bayesian computation.

We introduce a new representation of prior distributions over context trees that separates the specification of structural prior beliefs from the computational machinery used for Bayesian inference.
We establish the theoretical foundations of this representation, characterize the class of induced prior distributions, and show that several existing constructions, including CTW, BCT, and branching-process priors, arise naturally as particular instances of the proposed framework.
At the same time, the framework enables the construction of new classes of structurally interpretable priors that cannot be conveniently expressed using existing parametrizations.

Our second contribution is computational. We derive recursive algorithms for exact Bayesian inference under the proposed class of priors, including the computation of marginal likelihoods, exact posterior distributions, and maximum a posteriori trees.
An important consequence of this development is that exact marginal likelihoods become available for comparing competing prior specifications through Bayes factors, providing a principled methodology for both prior elicitation and model selection without sacrificing computational efficiency.

Finally, we investigate the practical implications of the proposed framework through simulation studies and a real-data application. The simulations assess the impact of different structural priors under controlled settings, while the financial market application illustrates how alternative prior specifications influence model evidence and predictive performance in practice.

The remainder of the paper is organized as follows.
Section \ref{sec-background} reviews Bayesian inference for context trees and discusses existing prior constructions.
Section \ref{sec-priors} introduces the proposed node-weighted representation and establishes its theoretical properties.
Section \ref{sec-inference} develops exact recursive algorithms for Bayesian inference under the proposed framework.
Section \ref{sec-applications} discusses practical uses of the proposed methods for model comparison via Bayes Factors and a procedure for selecting the maximal dependence range.
Sections \ref{sec-simulations} and \ref{sec-financial} present simulation studies and a financial market application, respectively.
Section \ref{sec-conclusion} concludes with a discussion of the proposed methodology and possible directions for future research. All proofs are given in \ref{appendix-A}.

\section{Bayesian Inference for Context Trees}\label{sec-background}

\subsection{Context Trees and Variable-Length Markov Chains}

Let $\mathcal{A} = \{0, 1, \dots, m-1\}$ be a finite alphabet with $m = |\mathcal{A}|$ symbols. We denote a string of symbols by $\mathbf{z}_{n_1}^{n_2} = z_{n_1} \dots z_{n_2} \in \mathcal{A}^{n_2 - n_1 + 1}$ and its length by $\ell(\mathbf{z}_{n_1}^{n_2}) = n_2 - n_1 + 1$, adopting the convention $\mathcal{A}^0 = \{\lambda\}$, where $\lambda$ is the empty string ($\ell(\lambda)=0$). The concatenation of two strings $\mathbf{z}$ and $\mathbf{y}$ is written as $\mathbf{zy}$. A string $\mathbf{s}$ of length $k$ is a \textit{suffix} of $\mathbf{z}_{1}^n$ if $\mathbf{s} = z_{n-k+1} \dots z_n$. A set of strings $\tau$ is \textit{proper} if no string in $\tau$ is a suffix of any other string in $\tau$. Moreover, a set $\tau$ is \textit{complete} if for any string $\mathbf{z}_1^n$ with $n \ge L+1$, there exists a suffix in $\tau$ that is a suffix of $\mathbf{z}_1^n$ \citep{CTW}.

The structure of a rooted tree \citep{gross2005graph} can be described via strings. By labeling the nodes using symbols from the alphabet $\mathcal{A}$, we can represent each node of the tree by the string corresponding to the path from that node to the root $\lambda$. In a \textit{full rooted tree}, each inner node $\mathbf{s}$ has exactly $m$ children, represented via concatenation as $k\mathbf{s}$ for $k \in \mathcal{A}$. Such a tree is uniquely characterized by its set of leaves (terminal nodes) $\tau$, which satisfies both the properness and completeness properties.

Throughout this work, every tree is a full rooted tree represented by its set of leaves $\tau$. We denote the set of all nodes of a tree by $\mathcal{N}(\tau)$ and the set of inner nodes by $\tau^c = \mathcal{N}(\tau) \setminus \tau$. The \textit{depth} of a tree is $\ell(\tau) = \max\{\ell(\mathbf{s}): \mathbf{s} \in \tau\}$. For a fixed maximal depth $L > 0$, we define the space of admissible trees as $\mathcal{T}_L = \{\tau: \ell(\tau) \leq L\}$. Any tree $\tau \in \mathcal{T}_L$ is a sub-tree of the \textit{maximal tree} $\tau_{\text{max}}$, whose leaves all have length $L$. 

By completeness, for any sequence $\mathbf{z}_1^n$ with $n \geq L+1$, there is a suffix in $\tau$. Since $\tau$ is proper, this suffix is unique; that is, given $\mathbf{z}_1^n$, there exists a unique $\mathbf{s} \in \tau$ that is a suffix of $\mathbf{z}_1^n$. This defines the \textit{suffix mapping function}
\begin{align*}
    \eta_\tau: \bigcup_{i=L+1}^\infty \mathcal{A}^i \rightarrow \tau,
\end{align*}
which maps $\mathbf{z}_{1}^{n}$ to its unique suffix $\mathbf{s} \in \tau$.

\begin{definition}
A \textit{context tree} is a pair $(\tau, \boldsymbol{\theta})$, where $\tau \in \mathcal{T}_L$ is a set of contexts (leaves) and
$$
\boldsymbol{\theta}=\{\boldsymbol{\theta}_{\mathbf{s}}:\mathbf{s}\in\tau\}
$$
is a collection of probability vectors on the $m$-simplex $\Delta_m$,
$\boldsymbol{\theta}_{\mathbf{s}} = \big[\theta_{\mathbf{s}}(0), \dots, \theta_{\mathbf{s}}(m-1)\big] \in \Delta_m$ such that $\sum_{k=0}^{m-1} \theta_{\mathbf{s}}(k) = 1$ and $\theta_{\mathbf{s}}(k) \ge 0$ for all $\mathbf{s} \in \tau$ and $k \in \mathcal{A}$.
\end{definition}

\begin{figure}[hb]
    \centering
    \begin{tikzpicture}[
    level distance=1.5cm,
    level 1/.style={sibling distance=2.5cm},
    level 2/.style={sibling distance=1.5cm},
    grow'=down,
    scale=0.6,
    edge from parent/.style={draw, thick},
    internal/.style={
        draw,
        circle,
        thick,
        fill=white,
        font=\small,
        align=center,
        minimum size=1pt,
        inner sep= 1.5pt
    },
    leaf/.style={
        draw,
        circle,
        thick,
        fill=black,
        font=\small,
        text=white,
        align=center,
        minimum size=1pt,
        inner sep=1.5pt
    }
]

\node[internal] (root) {$\lambda$}
    child { node[internal] {1}
        child { node[leaf] {1} }
        child { node[leaf] {0} }
    }
    child { node[leaf] {0} };

     \begin{scope}[every node/.style=    {text=black, font=\small}]

        \node at (root-1-1) [below=8pt] {$\boldsymbol{\theta}_{11}$};

        \node at (root-1-2) [below=8pt] {$\boldsymbol{\theta}_{01}$};

        \node at (root-2) [below=8pt] {$\boldsymbol{\theta}_{0}$};

    \end{scope}

\end{tikzpicture}
    \caption{Example of a context tree $(\tau, \boldsymbol{\theta})$ for a binary alphabet. Filled nodes represent contexts $\tau = \{0, 01, 11\}$, and each $\boldsymbol{\theta}_{\mathbf{s}}$ specifies the transition probabilities conditioned on context $\mathbf{s}$. Unfilled nodes represent inner nodes $\tau^c = \{\lambda, 1\}$, illustrating the underlying tree structure.}
    \label{contexttree}
\end{figure}

Variable-length Markov chains (VLMCs) generalize standard Markov chains by allowing the memory length to depend on the history via these contexts \citep{Buhl}. The intuition is that the conditional distribution of the next symbol depends on a suffix of the past whose length is itself data-dependent. The context tree encodes the relevant suffixes together with their associated transition probability vectors.

\begin{definition}
A stationary process $\mathbf{Z}=(Z_t)_{t=1}^n$ taking values in $\mathcal{A}$ is a VLMC compatible with the context tree $(\tau,\boldsymbol{\theta})$ if, for all $L<t\le n$ and $\mathbf{z}_{1}^{t-1}\in\mathcal{A}^{t-1}$,
\begin{align}
\label{vlmc}
    p(Z_t=k\mid \mathbf{Z}_{1}^{t-1}=\mathbf{z}_{1}^{t-1})
    =
    \theta_{\eta_\tau(\mathbf{z}_{1}^{t-1})}(k).
\end{align}
\end{definition}

By treating the first $L$ symbols of a realization $\mathbf{z}$ as fixed initial conditions, the likelihood under the VLMC model is
\begin{align}
\label{likeli}
p(\mathbf{z}\mid\tau,\boldsymbol{\theta})
=
\prod_{\mathbf{s}\in\tau}
\prod_{k\in\mathcal{A}}
\theta_{\mathbf{s}}(k)^{c_{\mathbf{s}k}(\mathbf{z})},
\end{align}
where $c_{\mathbf{s}k}(\mathbf{z})
=
\sum_{t=L+1}^{n}
\mathbf{1}
\left\{
z_t=k,\,
\eta_\tau(\mathbf{z}_1^{t-1})=\mathbf{s}
\right\}$ counts the number of occurrences of symbol $k$ following context $\mathbf{s}$ in the sequence $\mathbf{z}$.

\subsection{Bayesian Formulation}

We now introduce a Bayesian formulation for context trees. In this framework, both the tree structure $\tau$ and the transition probability vectors $\boldsymbol{\theta}$ are treated as unknown quantities and assigned prior distributions. The resulting hierarchical model is
\begin{align*}
    \tau &\sim \pi(\tau), \hspace{4.6cm} \tau \in \mathcal{T}_L, \\ 
    \boldsymbol{\theta} \mid \tau &\sim  \pi(\boldsymbol{\theta} \mid \tau), \hspace{4.088cm} \boldsymbol{\theta} \in \Delta_m^{|\tau|}, \\ 
    \mathbf{Z} \mid \tau, \boldsymbol{\theta} &\sim p(\mathbf{z} \mid \tau, \boldsymbol{\theta}), \hspace{3.795cm} \mathbf{z} \in \mathcal{A}^n.
\end{align*}
where the likelihood is given by \eqref{likeli}.

Throughout this work, our primary interest is the tree structure $\tau$. The transition probability vectors are therefore regarded as nuisance parameters and integrated out analytically. This results in the marginal likelihood of the observed sequence under a given tree,
\begin{align}
\label{eq:marginal}
p(\mathbf{z}\mid\tau)
=
\int
p(\mathbf{z}\mid\tau,\boldsymbol{\theta})
\pi(\boldsymbol{\theta}\mid\tau)
\,d\boldsymbol{\theta}.
\end{align}
Applying Bayes' theorem, the posterior distribution over context trees is
\begin{align}
\label{eq:posterior}
\pi(\tau\mid\mathbf{z})
=
\frac{\pi(\tau)\,p(\mathbf{z}\mid\tau)}
{p(\mathbf{z})},
\end{align}
where
\begin{align}
\label{eq:evidence}
p(\mathbf{z})
=
\sum_{\tau\in\mathcal{T}_L}
\pi(\tau)\,
p(\mathbf{z}\mid\tau)
\end{align}
is the marginal likelihood (or evidence) of the observed sequence.

A standard choice for the conditional prior on the transition probabilities is to assign independent symmetric Dirichlet distributions to the probability vector associated with each context. Specifically,
\begin{align}
\label{eq:dirprior}
\pi_\alpha(\boldsymbol{\theta}\mid\tau)
=
\prod_{\mathbf{s}\in\tau}
\frac{\Gamma(m\alpha)}
{\Gamma(\alpha)^m}
\prod_{k=0}^{m-1}
\theta_{\mathbf{s}}(k)^{\alpha-1},
\end{align}
where $\alpha>0$ is a hyperparameter. This specification ensures an efficient computation of the marginal likelihood of the observed sequence under a given tree, as the product of Dirichlet distributions conjugates to the likelihood function.

\begin{lemma}
\label{lem:marginal}
Under the prior \eqref{eq:dirprior} with hyperparameter $\alpha>0$, the marginal likelihood of $\mathbf{z}$ under tree $\tau$ is given by
\begin{align}\label{eq:margtree}
p_\alpha(\mathbf{z}\mid\tau)
=
\prod_{\mathbf{s}\in\tau}
\frac{\Gamma(m\alpha)}
{\Gamma(\alpha)^m}
\frac{
\prod_{k=0}^{m-1}
\Gamma\!\left(c_{\mathbf{s}k}(\mathbf{z})+\alpha\right)
}{
\Gamma\!\left(
\sum_{k=0}^{m-1}
c_{\mathbf{s}k}(\mathbf{z})
+m\alpha
\right)
}.
\end{align}
\end{lemma}

The marginal likelihood \eqref{eq:evidence} plays a central role in Bayesian inference. Besides defining the posterior distribution, it provides the basis for Bayesian model comparison through Bayes factors and quantifies the predictive probability assigned by the model to the observed sequence. Through the well-known connection between Bayesian prediction and universal source coding, its logarithm is also directly related to data compression. Consequently, the ability to compute the marginal likelihood is of fundamental importance.

The main difficulty lies in the fact that the expression in \eqref{eq:evidence} involves a sum over all trees in $\mathcal{T}_L$, whose cardinality grows doubly exponentially with the maximal depth ($O(m^{m^L})$). As a result, direct computation becomes intractable even for moderate values of $L$. In fact, the number of trees in $\mathcal{T}_L$ satisfies the recurrence relation $|\mathcal{T}_L| = |\mathcal{T}_{L-1}|^m + 1$, with $|\mathcal{T}_0| = |\{\lambda\}| = 1$. Then, for a simple binary alphabet $(|\mathcal{A}| = 2)$ and $L = 10$, the number of trees is approximately $|\mathcal{T}_{10}| \approx 1.44 \times 10^{217}$.

Essentially, the choice of the prior distribution $\pi(\tau)$ is intimately connected to the tractability of the marginal likelihood computation. In what follows, we review existing prior distributions and explain how different formulations address this computational challenge.

\subsection{Existing Prior Distributions}

The specification of the prior distribution $\pi(\tau)$ over the space of trees $\mathcal{T}_L$ remains the key modeling component. Since the posterior distribution, the marginal likelihood and, consequently, most Bayesian model comparison depend critically on this prior, its definition has a fundamental impact on inference.
Several approaches have been proposed in the literature, differing primarily in how structural regularization is imposed on $\tau$.

A first class of models is based on \textit{Context Tree Weighting} (CTW) and related constructions \citep{CTW}. In these approaches, the prior over trees is implicitly defined through a recursive mixture over local branching decisions. The resulting distribution can be interpreted as being induced by a stochastic process that generates trees top-down, where each node independently decides whether to stop or to branch into its $m$ children with probability $1/2$.
This class has the advantage of allowing efficient recursive computation of the marginal likelihood, but it restricts the form of admissible priors to those compatible with a fixed local branching mechanism.

A closely related family is given by Bayesian extensions of CTW, referred to as \textit{Bayesian Context Trees} (BCT) \citep{kontoyiannis2022bayesiancontexttreesmodelling}.
These methods also rely on recursive decompositions of the tree space induced by a branching process with a fixed splitting probability that can be different from $1/2$. In this setting, the prior on $\tau$ can be viewed as the distribution of a Galton-Watson-type process constrained to the finite-depth space $\mathcal{T}_L$ \citep{10.1214/23-BA1362}. While this yields exact and computationally efficient inference, the resulting prior is determined by a single parameter and cannot accommodate more detailed structural preferences.

\citet{Nakahara_2022} propose a class of distributions for rooted trees induced by branching-type processes in which the probability of splitting or stopping may depend on the node itself, allowing for greater heterogeneity across the tree structure. While this added flexibility enables richer structural modeling, the resulting prior still retains a locally recursive form, which is crucial for computational tractability but limits the ability to encode global structural preferences directly.

A more flexible alternative is to assign a prior directly over $\mathcal{T}_L$ without enforcing a recursive factorization.
For instance, \citet{VictorNancy} consider general priors over the space of context trees and perform posterior inference via Markov chain Monte Carlo methods. This approach allows for greater modeling flexibility, but inference relies on simulation-based methods and does not provide closed-form expressions for the marginal likelihood or exact recursive computation.

More generally, all existing approaches can be interpreted as defining a joint model similar to the one described in the previous section,
but they differ in the structural restrictions imposed on $\pi(\tau)$. In particular, tractable methods rely on strong factorizability assumptions, while more general priors sacrifice computational tractability.
This tension between structural flexibility and computational efficiency motivates the need for alternative constructions that allow richer specifications of $\pi(\tau)$ while still preserving exact or recursive computation of the marginal likelihood, as developed in the following sections.

\subsection{Structural Constraints and Prior Expressiveness}

A central challenge in specifying priors over context trees lies in the inherent complexity of the space $\mathcal{T}_L$, which encodes combinatorial and hierarchical structures. While existing approaches provide computationally convenient constructions, they differ in the extent to which they allow direct control over structural properties of the tree.

Most standard approaches define $\pi(\tau)$ indirectly through local probabilistic mechanisms, typically associated with node-wise branching decisions. In such constructions, the global shape of the tree emerges from decisions at each node, rather than being specified through global structural constraints.
As a consequence, it is difficult to directly encode prior beliefs about global characteristics of the tree, such as its overall depth distribution, sparsity patterns, or renewal-type behavior. Instead, such properties are only implicitly induced by the choice of local branching probabilities, often in a manner that is difficult to interpret or calibrate.

This construction based on local decisions is closely related to the recursive structure that enables efficient computation of marginal likelihoods. As a consequence, substantially different global tree structures may arise from very similar local branching specifications, limiting the ability to control global properties of the model.

These limitations suggest a gap between the way structural assumptions are naturally formulated in applications and the way they are encoded in existing Bayesian constructions. The next section introduces a framework that addresses this gap by providing a more direct representation of structural information on context trees.

\section{Node-Weighted Structural Priors}\label{sec-priors}

While the parametrization based on branching probabilities enable efficient recursive computation, it provides only indirect control over the structural properties of the induced distribution.
As a consequence, expressing prior beliefs about tree characteristics often requires reasoning about the cumulative effect of many local branching decisions.

Our approach is based on, instead of specifying how a tree is generated, specifying how probable each tree should be.
Since every context tree is uniquely determined by its set of terminal contexts, it is natural to associate a non-negative weight with each possible context and let the plausibility of a tree be determined by the weights of the contexts it contains.
In order to achieve this, we assign a local weight to every node and obtain the score of a tree by combining the weights of its leaves. Trees containing highly weighted contexts receive larger scores, whereas trees containing less desirable contexts are penalized. After normalization, these scores define a probability distribution over the space of admissible context trees.

This representation has two important advantages. First, it allows structural prior information to be specified directly in terms of the contexts themselves, rather than through an implicit branching mechanism. Second, as we show in Section \ref{sec-inference}, it preserves the recursive structure that makes exact Bayesian inference computationally feasible.

\subsection{Context-Tree Functions}

Let $w:\mathcal{N}(\tau_{\max}) \rightarrow \mathbb{R}_0^+$
be a non-negative node-weight function assigning a weight to every node of the maximal tree.
Given a context tree $\tau$, we define its \emph{tree score} as
$$
W(\tau)=\prod_{\mathbf{s}\in\tau}w(\mathbf{s}).
$$
The score measures the structural plausibility of the tree according to the chosen node weights. Larger values indicate structures that are considered more plausible a priori.
Normalizing these scores over the space $\mathcal T_L$ yields the prior distribution
\begin{align} \label{eq:treeprior}
    \pi_W(\tau)=
\frac{W(\tau)}
{\sum_{\tau'\in\mathcal T_L}W(\tau')}.
\end{align}
We refer to any function $W$ admitting this factorization as a \emph{context-tree function}. Throughout the paper, context-tree functions constitute the basic representation used to construct prior distributions over context trees.

\subsection{Examples}

One of the main advantages of the proposed representation is that a wide variety of structurally different priors can be obtained simply by modifying the node-weight function. Rather than developing a new prior from scratch, one can specify only the local weights, from which the corresponding distribution over context trees follows automatically. 

Our examples of weight functions can be organized into three general classes: those based on constant weights, those based on indicator constraints, and those based on node length. Table~\ref{table_ex} summarizes several representative examples. The notation and interpretation of the resulting prior distributions are explored in detail below. Throughout, we use the same function name for the weight function $w$ and for the corresponding context-tree function $W$. 

\renewcommand{\arraystretch}{1.7}
\begin{table}[h!]
    \centering
    \small{
    \caption{Examples of weight functions and corresponding tree scores} 
    \label{table_ex}
    \vspace{0.2cm}
    \begin{tabular}{c|c|c}
    \toprule
    \textbf{Function name} & $w(\mathbf{s})$ & $W(\tau)$ \\ \midrule
    Unity  & $1$ & $1$ \\ \hline
    $\beta$-constant & $\beta$ & $\beta^{|\tau|}$ \\ \hline
    $a$-renewal indicator & 
    $\begin{cases}
        0 & \text{if } s_j = a \text{ for some } j = 1, \dots t-1 \\
        1 & \text{otherwise}.
    \end{cases}$ & $\mathbf{1}\{\tau \in \mathcal{T}_L^a\}$ \\ \hline
     $l$-$m$-depth indicator & $\mathbf{1}\{l \leq \ell(\mathbf{s}) \leq m\}$ & $\mathbf{1}\{\tau \in \mathcal{T}_m \setminus \mathcal{T}_l\}$ \\ \hline
    $g$-length-based & $g(\ell(\mathbf{s}))$ & $\displaystyle{\prod_{\mathbf{s} \in \tau} g(\ell(\mathbf{s}))}$ \\ \hline
    $\beta$-exponential & $\exp(\beta \ell(\mathbf{s}))$ & $\displaystyle{\prod_{\mathbf{s} \in \tau} \exp(\beta \ell(\mathbf{s}))}$ \\ \hline
    CTW  & 
    $\begin{cases}
        1/4 & \text{if } \ell(\mathbf{s}) < L \\
        1/2 & \text{if } \ell(\mathbf{s}) = L.
    \end{cases}$ & $(1/2)^{|\tau| + M(\tau)}$ \\ \hline
    $\beta$-BCT  & 
    $\begin{cases}
        (1-\beta)^{1/(m-1)}\beta \quad &\text{if} \ \ell(\mathbf{s})<L \\
        (1-\beta)^{1/(m-1)} & \text{if} \ \ell(\mathbf{s})=L. 
    \end{cases}$ & $(1-\beta)^{\frac{|\tau|}{m-1}} \beta^{M(\tau)}$ \\ \hline
    $\beta$-target $l$-depth  & $\beta^{-|\ell(\mathbf{s}) - l|}$ & $\displaystyle{\prod_{\mathbf{s} \in \tau} \beta^{-|\ell(\mathbf{s}) - l|}}$ \\
    \bottomrule
    \end{tabular}}
\end{table}

The simplest example, the unity function, is characterized by assigning unit weights to every node, producing the uniform distribution over $\mathcal{T}_L$. A generalization of this is obtained by setting a constant weight $\beta > 0$, yielding the $\beta$-constant function. Choosing weights smaller than one ($\beta < 1$) penalizes large trees and favors simpler structures, while weights greater than one ($\beta > 1$) favor deeper trees.

More elaborate constructions can also be made, such as constraints based on the path of the nodes. For instance, the $a$-renewal indicator function sets non-zero weights only for nodes that do not contain the symbol $a \in \mathcal{A}$ internally. Consequently, it induces a uniform distribution on the subset of $a$-renewing trees $\mathcal{T}_L^a$ (see \cite{VictorNancy} for details).
As an example, \cite{Galves_2012} uses VLMC to model linguistic rhythm in written texts,
considering a specific symbol to represent the end of sentences as a renewal symbol,
what allows them to obtain independent blocks of symbols to be used in a bootstrap scheme.
In a Bayesian construction, this assumption could be directly encoded in prior distribution through the $a$-renewal indicator function.

Another type of possible restriction is based on the length of the contexts (or depth of the trees). The $l$-$m$-depth indicator function assigns non-zero weights only to nodes $\mathbf{s}$ whose lengths satisfy $0 \le l \le \ell(\mathbf{s}) \le m \le L$. This produces a uniform distribution on $\mathcal{T}_m \setminus \mathcal{T}_l$, excluding trees whose depths fall outside the specified range,
which may be useful in applications where the relevant memory length is known to lie within a specific interval.

More generally, we can also assign different weights based on node length. That is, for any function $g: \mathbb{Z}_0^+ \to \mathbb{R}_0^+$, the $g$-length-based function sets the weight of a node $\mathbf{s}$ according to $g$ evaluated at its length $\ell(\mathbf{s})$. Examples of such functions include the $\beta$-exponential ($\beta \in \mathbb{R}$), the CTW, the $\beta$-BCT ($\beta \in (0, 1)$), where $M(\tau)$ denotes the number of contexts with depth strictly less than $L$, and the $\beta$-target $l$-depth ($\beta > 0$, $0 \le l \le L$). 

The CTW and $\beta$-BCT functions induce priors already established in the literature \citep{CTW, kontoyiannis2022bayesiancontexttreesmodelling}, and correspond to special cases of distributions generated by branching processes (see Subsection \ref{sec:branching}). The $\beta$-target $l$-depth deserves special attention. It is constructed to ensure that the resulting distribution favors trees whose contexts have lengths close to $l$. The parameter $\beta$ controls how much less likely trees with depths differing from $l$ are. Larger values of $\beta$ lead to distributions that are highly concentrated around the maximal tree of depth $l$. Therefore, its application may be useful for problems that exhibit a known cyclic pattern, such as patterns that are likely to repeat every $l$ steps.

These examples show that the proposed framework unifies a broad class of existing constructions while substantially simplifying the specification of new structural priors.

\subsection{Properties}

In this subsection we describe some direct properties of context-tree functions.

\begin{property}[Product]
     The class of context-tree functions is closed under finite products. Indeed, if $w_1$ and $w_2$ are node-weight functions, then defining $w(\mathbf{s}) = w_1(\mathbf{s}) w_2(\mathbf{s})$ induces a context-tree function with associated tree score $W(\tau) = W_1(\tau) W_2(\tau)$.
\end{property}

However, note that if a product of functions is a context-tree function, it does not necessarily follow that each factor is itself a context-tree function. For instance, consider the constant functions $F(\tau) = 1/2$ and $G(\tau) = 2$. Their product $(FG)(\tau) = 1$ is a context-tree function (the unity), but $F$ and $G$ are not.


\begin{property}[Rescaling]
    If $w'(\mathbf{s}) = c\, w(\mathbf{s})$ for some constant $c>0$, then $W'(\tau) = c^{|\tau|} W(\tau)$.
\end{property}

Therefore, uniform rescaling of node weights modifies the induced distribution over trees through a factor that depends on the number of leaves. This implies that the scale of $w$ plays an essential role in controlling the preference over tree sizes, with $c>1$ favoring larger trees and $c<1$ favoring smaller ones. 

\begin{property}[Non-uniqueness]
    The node-weight representation $\pi_W$ of a distribution is not unique. 
\end{property}

Distinct node-weight functions may induce the same distribution over $\mathcal{T}_L$ when they differ by transformations that preserve relative tree scores. For instance, consider $\mathcal{T}_1$, which contains only two trees: the tree with a single root node and the fully expanded tree with $m$ leaves. If $w(\lambda) = c$ and $w(k) = c^{1/m}$ for all children $k \in \mathcal{A}$, then the induced distribution is uniform over $\mathcal{T}_1$, a property that is also obtained by the unity function $w'(\mathbf{s}) = 1$.

\subsection{Connection with Branching Processes} \label{sec:branching}

Prior distributions based on context-tree functions $\pi_W$ are closely related to distributions induced by branching processes, as both construct tree probabilities through multiplicative node-level operations. In the latter, however, the distribution is described by a collection of \textit{branching probabilities} that characterize node-branching decisions, rather than by a node-weight function.

More precisely, consider a collection $\{b_{\mathbf{s}}: \mathbf{s} \in \mathcal{N}(\tau_{\text{max}})\}$, where $b_{\mathbf{s}} \in [0, 1]$ represents the \textit{branching probability} of a node $\mathbf{s}$, i.e. the probability that $\mathbf{s}$ have children. Then, the mass assigned to a tree $\tau \in \mathcal{T}_L$ is
\begin{align*}
    \pi_{\{b_{\mathbf{s}}\}}(\tau) = \prod_{\mathbf{s} \in \tau} (1-b_{\mathbf{s}}) \prod_{\mathbf{s}^\prime \in \tau^c} b_{\mathbf{s}^\prime}.
\end{align*}
In simple terms, the probability of a tree is the product of the branching probabilities associated with the nodes that branch and the stopping probabilities associated with the nodes that do not branch. It is shown in \cite{Nakahara_2022} that this construction in fact defines a probability distribution on $\mathcal{T}_L$. For instance, the $\beta$-BCT distribution is a special case with branching probability $b_{\mathbf{s}} = (1-\beta)$ for all $\mathbf{s} \in \mathcal{N}(\tau_{\text{max}})$.

The connection between $\pi_W$ and $\pi_{\{b_{\mathbf{s}\}}}$ lies in the fact that they are distinct parameterizations of the same probability distribution. In particular, for every context-tree function $W$, one can construct a unique set of branching probabilities $\{b_{\mathbf{s}}\}$ that yields $\pi_W$. Conversely, any collection of branching probabilities specifies a context-tree function whose induced distribution coincides with $\pi_{\{b_{\mathbf{s}}\}}$.

To establish this equivalence, first, given a context-tree function $W$, for $\mathbf{s} \in \mathcal{N}(\tau_{\text{max}})$, we define the branching probability $b_\mathbf{s}$ such that
\begin{align}
\label{branch}
(1-b_{\mathbf{s}}):=\frac{w(\mathbf{s})}{\Sigma_W(\mathbf{s})},
\end{align}
where $\Sigma_W(\mathbf{s})$ denotes the sum of the contributions of the function $W$ over all sub-trees rooted at the node $\mathbf{s}$. Function $\Sigma_W$ will be more detailed in Subsection \ref{subsec:recur}. The choice of the ratio in \eqref{branch} is natural, since it represents the relative contribution of the tree that stops at $\mathbf{s}$ with respect to the total contribution of all possible sub-trees below $\mathbf{s}$, and hence can be interpreted as the probability of terminating the branching process at node $\mathbf{s}$. 

On the other hand, given a collection of branching probabilities $\{b_{\mathbf{s}}: \mathbf{s} \in \mathcal{N}(\tau_{\text{max}})\}$, we consider the weight function as
    \begin{align} \label{carac-naka}
    w^\prime(\mathbf{s}) := (1-b_{\mathbf{s}}) \prod_{\mathbf{s}^\prime \in A(\mathbf{s})} (b_{\mathbf{s}^\prime})^{m^{\ell(\mathbf{s}^\prime) - \ell(\mathbf{s})}},
    \end{align}
where $A(\mathbf{s})$ represents the set of ancestor nodes of  node $\mathbf{s}$. The weight assigned to a node $\mathbf{s}$ in \eqref{carac-naka} combines the stopping probability of $\mathbf{s}$ with the contributions of the branching probabilities of its ancestors (the inner nodes on the path from the root to $\mathbf{s}$), accounting for their multiplicities through the depth of $\mathbf{s}$ and the size of the alphabet considered. 

This construction establishes the equivalence of the induced tree distributions. For any $\tau \in \mathcal{T}_L$, we have:
\begin{equation*}
\pi_W(\tau) = \frac{W(\tau)}{\sum_{\tau^\prime \in \mathcal{T}_L} W(\tau^\prime)} = \pi_{\{b_{\mathbf{s}}\}}(\tau) = \prod_{\mathbf{s} \in \tau} w^\prime(\mathbf{s}) = W^\prime(\tau) = \pi_{W^\prime}(\tau).
\end{equation*}

Note that this equivalence also illustrates the non-uniqueness property of $\pi_W$. The weight function defined in \eqref{carac-naka} may induce a context-tree function $W^\prime$ that differs from $W$, meaning that $\pi_W$ and $\pi_{W^\prime}$ represent the same prior distribution under different node-weight parameterizations.

\section{Exact Bayesian Inference}\label{sec-inference}

The preceding sections established the Bayesian formulation for context trees and introduced context-tree functions as a flexible representation of prior distributions. We now turn to the computational consequences of this representation.

The key observation is that the Dirichlet-integrated likelihood
factorizes over the leaves of a context tree, just as the prior weight function does. Consequently, their product is again a context-tree function. This factorization allows sums over the complete tree space $\mathcal{T}_L$ to be evaluated recursively, without enumerating its elements. In particular, the proposed framework provides exact recursive
procedures for computing the marginal likelihood, the posterior
distribution over trees, and a maximum a posteriori tree. These
procedures all rely on the same recursion, where sums are
used for marginal likelihoods and posterior normalization.

Under the symmetric Dirichlet prior in \eqref{eq:dirprior}, define, for
each node $\mathbf{s}\in\mathcal{N}(\tau_{\text{max}})$, its contribution to the
integrated likelihood as
\begin{align*}
    q_{\alpha}(\mathbf{s}, \mathbf{z}) = \frac{\Gamma(m \alpha)}{\Gamma(\alpha)^m} \frac{\prod_{k=0}^{m-1} \Gamma \big(c_{\mathbf{s}k}(\mathbf{z}) + \alpha\big)}{\Gamma \left(\sum_{k=0}^{m-1} c_{\mathbf{s}k}(\mathbf{z}) + m\alpha \right)}.
\end{align*}
It follows that $Q_\alpha(\tau, \mathbf{z}) = \prod_{\mathbf{s} \in \tau} q_{\alpha}(\mathbf{s}, \mathbf{z}) = p_\alpha(\mathbf{z} \mid \tau)$ is itself a context-tree function.

The marginal likelihood can thus be expressed as:
\begin{align}\label{mod_evi}
    p(\mathbf{z}; W, \alpha) = \sum_{\tau \in \mathcal{T}_L} \pi_W(\tau) p_{\alpha}(\mathbf{z} \mid \tau) = \frac{\sum_{\tau \in \mathcal{T}_L} W(\tau) Q_{\alpha} (\tau, \mathbf{z}) }{\sum_{\tau^\prime \in \mathcal{T}_L} W(\tau^\prime)},
\end{align}
where $WQ_{\alpha}$ is a context-tree function by the product property. Furthermore, by Bayes' theorem, the posterior distribution over trees $\tau$ is given by:
\begin{align} \label{postdist}
    \pi_{W, \alpha}(\tau \mid \mathbf{z}) = \frac{\pi_W(\tau) p_{\alpha}(\mathbf{z} \mid \tau)}{p(\mathbf{z}; W, \alpha)} = \frac{W(\tau)Q_\alpha(\tau, \mathbf{z})}{\sum_{\tau^\prime \in \mathcal{T}_L} W(\tau^\prime) Q_\alpha(\tau^\prime, \mathbf{z})}.
\end{align}
That is, the posterior distribution is directly proportional to the product context-tree function $WQ_\alpha$.

The property above is similar to saying that $\pi_W(\tau)$ and $\pi_{W, \alpha}(\tau \mid \mathbf{z})$ are conjugate distributions with respect to the likelihood function. It also allows us to recover the branching-process representation of the posterior distribution through \eqref{branch}, generalizing the results of \cite{10.1214/23-BA1362} and producing a simple method for obtaining posterior samples.

Equations \eqref{mod_evi} and \eqref{postdist} establish a central theoretical result: evaluating key Bayesian quantities reduces to computing sums of context-tree functions over the tree space $\mathcal{T}_L$. In the following subsections, we demonstrate that these sums can be computed efficiently, preserving the computational tractability. Furthermore, we show that these representations naturally lead to a method for obtaining point estimators.

\subsection{Recursive Computation} \label{subsec:recur}

Over the years, several methods have been developed to compute sums of functions defined on trees from quantities associated with their nodes. In the context of Bayesian coding, examples include the CTW algorithm proposed by \cite{CTW} and the algorithm of \cite{394633}, while a more recent approach in terms of probability distributions on tree space was proposed by \cite{Nakahara_2022}. In Theorem \ref{theo1} we demonstrate that the recursive mechanism underlying these methods can be used to compute sums over $\mathcal{T}_L$ of general context-tree functions.  

\begin{theorem} \label{theo1}
    Let $W: \mathcal{T}_L \rightarrow \mathbb R$ be a context-tree function. Then, the summation:
    \begin{align*}
        \sum_{\tau \in \mathcal{T}_L} W(\tau) = \sum_{\tau \in \mathcal{T}_L} \prod_{\mathbf{s} \in \tau} w(\mathbf{s}),
    \end{align*}
    can be obtained by the following procedure:
    \begin{enumerate}
        \item Starting at the leaves and proceeding towards the root $\lambda$: compute the value $\Sigma_W(\mathbf{s})$ for each node $\mathbf{s}\in\mathcal{N}(\tau_{\mathrm{max}})$ according to:
        \begin{align*}
            \Sigma_W(\mathbf{s}) = 
            \begin{cases}
            w(\mathbf{s}), & \text{if } \ell(\mathbf{s}) = L,\\[6pt]
            w(\mathbf{s}) + \displaystyle\prod_{k=0}^{m-1}\Sigma_W(k\mathbf{s}), & \text{if } \ell(\mathbf{s}) < L.
            \end{cases}
        \end{align*}
        \item After all nodes have been evaluated, the desired summation is obtained at the root $$\Sigma_W(\lambda) = \sum_{\tau\in\mathcal{T}_L} W(\tau).$$
    \end{enumerate}
\end{theorem}

The subscript in $\Sigma_{\bullet}$ identifies the context-tree function over which the sum is performed. For each $\mathbf{s} \in \mathcal{N}(\tau_{\text{max}})$, the quantity $\Sigma_W(\mathbf{s})$ can be interpreted as the sum of the values of the function $W$ over all trees rooted at $\mathbf{s}$ with depth at most $L-\ell(\mathbf{s})$ (see equation \eqref{eq-app} in \ref{appendix-A}). Therefore, the recursive procedure accumulates these partial sums as it moves upward through the tree towards the root $\lambda$. Since each node contributes to the computation exactly once and the number of nodes in the maximal tree is $|\mathcal{N}(\tau_{\text{max}})| = (m^{L+1}-1)/(m-1)$, the algorithm in Theorem~\ref{theo1} has computational cost $O(|\mathcal{N}(\tau_{\text{max}})|) = O(m^L)$.  In contrast, a brute-force computation of the sum requires evaluating every tree in $\mathcal{T}_L$ resulting in a complexity of $O(m^{m^L})$.

It is also possible to obtain an analogous result of Theorem \ref{theo1} for computing the maximum over all trees $\max_{\tau \in \mathcal{T}_L} W(\tau)$. This is achieved by making the following modification: instead of the function $\Sigma_W$, define $\Upsilon_W$ such that:
\begin{align} \label{eq:max}
    \Upsilon_W(\mathbf{s}) = 
    \begin{cases}
        w(\mathbf{s}), \ &\text{if} \ \ell(\mathbf{s}) = L, \\[6pt]
        \max\left\{w(\mathbf{s}),  \displaystyle{\prod_{k=0}^{m-1}} \Upsilon_W(k\mathbf{s}) \right\}, \ &\text{if} \ \ell(\mathbf{s}) < L, 
    \end{cases}
\end{align}
so as in the summation case, the maximum is obtained by evaluating $\Upsilon_W$ at the root $\Upsilon_W(\lambda) = \max_{\tau \in \mathcal{T}_L} W(\tau)$. 

\begin{example} \label{ex_summax}
    Let $\mathcal{A} = \{0,1\}$ and $L = 2$. Consider the unity function given by $w(\mathbf{s}) = 1$, for all $\mathbf{s} \in \mathcal{N}(\tau_{\text{max}})$. We use Theorem \ref{theo1} to compute both the sum and the maximum of $W(\tau) = 1$ over all trees in $\mathcal{T}_2$. Figure \ref{ex3} illustrates how the procedure works. 
        \begin{figure}[H]
            \centering
            \begin{tikzpicture}[
            level distance=1.5cm,
            level 1/.style={sibling distance=6cm},
            level 2/.style={sibling distance=3cm},
            grow'=down,
            scale=0.6,
            every node/.style={
                draw,
                circle,
                thick,
                fill=black,
                font=\small,
                align=center,
                text centered,
                minimum size=1pt,
                inner sep=1.5pt},
            edge from parent/.style={draw, thick},
            internal/.style={
                draw,
                circle,
                thick,
                fill=white,
                minimum size=1pt,
                inner sep=1.5pt},
                ]
            \node[internal] (root) {$\lambda$}
            child { node[internal] {$1$}
            child { node {\color{white}{$1$}} }
            child { node {\color{white}{$0$}} }
        }
        child { node[internal] {$0$}
            child { node {\color{white}{$1$}} }
            child { node {\color{white}{$0$}} }
        };

    \begin{scope}[every node/.style={text=black, font=\small}]
        \node at (root-1-1) [below=10pt] {$w(11)=1$};
        \node at (root-1-2) [below=10pt] {$w(01)=1$};
        \node at (root-2-1) [below=10pt] {$w(10)=1$};
        \node at (root-2-2) [below=10pt] {$w(00)=1$};
        \node at (root-1) [right=10pt] {$w(1)=1$};
        \node at (root-2) [left=10pt] {$w(0)=1$};
        \node at (root) [above=15pt] {$w(\lambda)=1$};
        \node at (root-1) [right=10pt, yshift=15pt] {$\Sigma_W(1)=2$};
        \node at (root-2) [left=10pt, yshift=15pt] {$\Sigma_W(0)=2$};
        \node at (root) [above=30pt] {$\Sigma_W(\lambda)=5$};
        \node at (root-1) [right=10pt, yshift=30pt] {$\Upsilon_W(1)=1$};
        \node at (root-2) [left=10pt, yshift=30pt] {$\Upsilon_W(0)=1$};
        \node at (root) [above=45pt] {$\Upsilon_W(\lambda)=1$};
    \end{scope}
    \end{tikzpicture}
    \caption{Recursive computation of the sum and maximum over all trees for the unity weight function. Values $w(\mathbf{s})$ are first computed for all nodes of maximal tree $\tau_{\text{max}}$. Then, proceeding from the leaves to the root, $\Sigma_W(\mathbf{s})$ and $\Upsilon_W(\mathbf{s})$ are computed recursively. Finally, the sum and maximum are obtained at the root as $\Sigma_W(\lambda)$ and $\Upsilon_W(\lambda)$. The values of $w$, $\Sigma_W$, and $\Upsilon_W$ are displayed next to their corresponding nodes in the figure.}
    \label{ex3}
\end{figure}
\end{example}

\subsection{Computing the Marginal Likelihood}

A direct consequence of Theorem \ref{theo1} is the exact and recursive computation of the marginal likelihood.

\begin{corollary}\label{cor:1}
    The marginal likelihood $p(\mathbf{z}; W, \alpha)$ from \eqref{mod_evi}
    can be obtained by the following procedure:
    \begin{enumerate}
        \item For each node $\mathbf{s}\in\mathcal{N}(\tau_{\mathrm{max}})$, compute the suffix counts $c_{\mathbf{s}k}(\mathbf{z})$;
        \item Starting at the leaves and proceeding towards the root $\lambda$: simultaneously compute the values of $\Sigma_W(\mathbf{s})$ and $\Sigma_{WQ_\alpha}(\mathbf{s})$ for each node $\mathbf{s}\in\mathcal{N}(\tau_{\mathrm{max}})$;
        \item After all nodes have been evaluated, the marginal likelihood is obtained by evaluating the $\Sigma$-functions at the root:
        \begin{align*}
            p(\mathbf{z}; W, \alpha) = \frac{\Sigma_{WQ_{\alpha}}(\lambda)}{\Sigma_W(\lambda)}.
        \end{align*}
    \end{enumerate}
\end{corollary}

Computing the sum for the context-tree function $WQ_\alpha$ is computationally more involved, as it requires evaluating the weight $q_\alpha(\mathbf{s}, \mathbf{z})$, which depends on the suffix counts $c_{\mathbf{s}k}(\mathbf{z})$. These counts can be precomputed in $O(nL)$ time by a single pass over the observed sequence. Therefore, in addition to the standard $O(m^L)$ cost of the recursive procedure in Theorem \ref{theo1}, the overall computational complexity for computing the marginal likelihood is $O(nL+m^L)$. By comparison, a brute-force computation of the marginal likelihood has complexity $O(nL+m^{m^L})$.

Corollary \ref{cor:1} can be viewed as an extension of the classic Context Tree Weighting (CTW) algorithm \citep{CTW} by allowing alternative \textit{weighting} strategies determined by the context-tree function $W$. Indeed, for a binary alphabet $(|\mathcal{A}|=2)$, when we choose the CTW function and set the Dirichlet hyperparameter to $\alpha = 1/2$, $p(\mathbf{z}; W, \alpha)$ corresponds exactly to the probability of the sequence assigned by the CTW. The only difference in our method is that the recursive procedure is carried out simultaneously for two context-tree functions: one corresponding to the product $WQ_\alpha$, and another corresponding to $W$.

Obtaining the posterior distribution \eqref{postdist} is even simpler, requiring only the sum of the context-tree function $WQ_\alpha$. In particular, it follows directly from the marginal likelihood computation in Corollary \ref{cor:1}:
\begin{align*}
    \pi_{W, \alpha}(\tau \mid \mathbf{z}) = \frac{W(\tau)Q_{\alpha}(\tau, \mathbf{z})}{\Sigma_{WQ_\alpha}(\lambda)}. 
\end{align*}


The exact posterior distribution allows us to evaluate the performance of MCMC methods that bypass the need for computing the marginal likelihood. In \ref{appendix-b}, we present an analysis of the Metropolis--Hastings algorithm proposed by \cite{kontoyiannis2022bayesiancontexttreesmodelling} on a simulated dataset, highlighting the statistical trade-offs incurred by adopting such an approach beyond the obvious loss of the standard Bayesian model selection.

\subsection{MAP Estimation}

From a Bayesian perspective, point estimation is commonly performed using the maximum a posteriori (MAP) estimator, which in our setting corresponds to a context tree that maximizes the posterior probability. From the form of posterior distribution \eqref{postdist}, a posterior mode tree, denoted by $\tau_{\text{map}}$, must satisfy
\begin{align} \label{eq:maxarg}
\tau_{\text{map}} 
                  = \arg\max_{\tau \in \mathcal{T}_L} \left( \frac{W(\tau)Q_\alpha(\tau, \mathbf{z})}{\sum_{\tau^\prime \in \mathcal{T}_L} W(\tau^\prime)Q_\alpha(\tau^\prime,\mathbf{z})} \right) 
                  = \arg\max_{\tau \in \mathcal{T}_L} \left( W(\tau) Q_\alpha(\tau, \mathbf{z}) \right).
\end{align}
Therefore, identifying a MAP estimator reduces to finding a tree that maximizes the context-tree function $WQ_\alpha$ over $\mathcal{T}_L$. 

The problem of searching a maximizing tree over the tree space for any context-tree function can also be solved recursively. In fact, as a consequence of Theorem \ref{theo1}, the ability to compute the maximum value of a context-tree function $W$ via the $\Upsilon$-function in \eqref{eq:max} leads to an algorithm for identifying a maximizing tree in simple steps. The underlying mechanism corresponds to that used in the Context Tree Maximizing (CTM) algorithm \citep{531122}.

\begin{proposition}\label{prop1}
    Let $W: \mathcal{T}_L \rightarrow \mathbb R$ be a context-tree function. Then a tree $\tau^\ast\in\mathcal{T}_L$ that satisfies
    $         
        W(\tau^\ast)= \max_{\tau\in\mathcal{T}_L} W(\tau) = \Upsilon_W(\lambda)
    $
    can be obtained by the following algorithm:
    \begin{enumerate}
        \item For each node $\mathbf{s} \in \mathcal{N}(\tau_{\mathrm{max}})$, compute the value of $\Upsilon_W(\mathbf{s})$. 
        \item Starting at the root $\lambda$ and proceeding towards the leaves: inspect each node $\mathbf{s}\in\mathcal{N}(\tau_{\mathrm{max}})$, prune the sub-tree below $\mathbf{s}$ if
            \begin{align} \label{max}
                w(\mathbf{s}) \ge \prod_{k=0}^{m-1}\Upsilon_W(k\mathbf{s});
            \end{align}
        Otherwise continue the same test for each child $k\mathbf{s}$, $k=0,\dots,m-1$.
        \item When all nodes have been inspected, the remaining tree is $\tau^\ast$.
    \end{enumerate}
\end{proposition}

Applying Proposition \ref{prop1} to the context-tree function $WQ_{\alpha}$, after computing the suffix counts $c_{\mathbf{s}k}(\mathbf{z})$, corresponds through \eqref{eq:maxarg} to finding a MAP tree $\tau_{\text{map}}$. This naturally extends the BCT algorithm proposed by \cite{kontoyiannis2022bayesiancontexttreesmodelling} to any prior distribution proportional to a context-tree function $W$. In particular, when $W$ is chosen as the $\beta$-BCT function, our method reduces to the original BCT algorithm.

\section{Methodological Applications}\label{sec-applications}

\subsection{Predictive Model Comparison and Bayes Factors}

A key motivation for the proposed framework is the availability of the exact marginal likelihood $p(\mathbf{z}; W, \alpha)$ for any context-tree function $W$, as established in Corollary~\ref{cor:1}. Given two competing prior specifications $W_1$ and $W_2$, the \textit{Bayes factor} of $W_1$ against $W_2$ is defined as
\begin{align} \label{eq:BF}
    \mathrm{BF}(W_1, W_2) = \frac{p(\mathbf{z}; W_1, \alpha)}{p(\mathbf{z}; W_2, \alpha)}.
\end{align}
Values $\log_{10} \mathrm{BF}(W_1, W_2) > 1$ are conventionally regarded as strong evidence in favor of $W_1$ \citep{Kass01061995}.
The Bayes factor compares the full predictive distributions of $\mathbf{z}$ under each prior, averaging over all trees rather than conditioning on any particular one, which makes it sensitive to how well the prior allocates probability mass globally.
A prior that concentrates mass on well-fitting trees will be preferred even if both priors share the same MAP tree.

To understand what is gained from selecting a prior with higher marginal likelihood, it is useful to note that $p(\mathbf{z}; W, \alpha)$ is the predictive probability of the observed sequence under the model,
\begin{align*}
    p(\mathbf{z}; W, \alpha) = \prod_{t=L+1}^{n} p(z_t \mid z_1, \ldots, z_{t-1}; W, \alpha),
\end{align*}
where each factor is the one-step-ahead predictive probability of $z_t$ given the past, obtained by averaging the transition probabilities over the posterior distribution of trees given $z_1, \ldots, z_{t-1}$.
A model with higher marginal likelihood assigns higher probability to the sequence that was actually observed, meaning it provides better sequential predictions on average.
Consequently, selecting the prior with the highest marginal likelihood is equivalent to selecting the model whose structural assumptions lead to the best predictive performance on the observed data.

A concrete example of this connection is in terms of lossless data compression. By the source coding theorem, $-\log_2 p(\mathbf{z}; W, \alpha)$ equals the codelength (in bits) that an arithmetic coder based on model $(W, \alpha)$ assigns to $\mathbf{z}$ \citep{CTW}.
A prior with higher marginal likelihood therefore compresses the data more efficiently. More precisely, the gain in bits from using $W_1$ instead of $W_2$ to encode $\mathbf{z}$ is exactly $\log_2 \mathrm{BF}(W_1, W_2)$.
Incorporating correct structural knowledge about the process, such as the relevant context length or the absence of certain dependencies, directly translates into compression gains, while a misspecified prior that spreads probability mass over structurally implausible trees wastes bits encoding events that the data do not support.

Comparisons of practical interest include different weight functions within the same family (e.g., $\beta$-target $l$-depth for varying $\beta$ or $l$), priors encoding different structural constraints, or established constructions such as CTW and $\beta$-BCT versus newly proposed ones.
In existing approaches such as BCT \citep{kontoyiannis2022bayesiancontexttreesmodelling}, the marginal likelihood is restricted to the family of branching-process priors, limiting Bayes factor comparisons to variations of a single branching parameter.
The proposed framework removes this restriction: since the marginal likelihood is exact and tractable for any context-tree function $W$, Bayes factors between priors with entirely different structural forms are available.

\subsection{Maximal Depth Selection} \label{subsec:max_depth}

A key modeling choice in the proposed framework is the maximal depth $L$, which determines
the space of admissible trees $\mathcal{T}_L$. Although the recursive algorithms in Section
\ref{sec-inference} are efficient for fixed $L$, the choice of $L$ itself affects inference:
a value that is too small may prevent the model from capturing genuine long-range dependencies,
while a value that is unnecessarily large expands the model space, increasing the computational complexity and leading to a more diffuse posterior distribution.

A natural Bayesian treatment assigns a prior distribution $p(L)$ over the candidate depths
$L = 1, 2, \dots, L_{\max}$ and updates it with the data.
Since for each fixed $L$ the proposed framework provides the exact marginal likelihood $p(\mathbf{z}; W, \alpha)$
via Corollary \ref{cor:1}, the posterior distribution of $L$ is available in closed form:
\begin{align} \label{eq:postL}
    p(L \mid \mathbf{z}) \propto p(\mathbf{z}; W, \alpha, L)\, p(L),
\end{align}
where we write $p(\mathbf{z}; W, \alpha, L)$ to make the dependence on the depth explicit.

Because the marginal likelihood must be evaluated for all candidate depths in order to identify the MAP estimate $L^* = \arg\max_L p(L \mid \mathbf{z})$, depth selection reduces to running the recursive algorithm of Corollary~\ref{cor:1} once for each value of $L \in \{1, \dots, L_{\max}\}$ and selecting the depth that maximizes $p(\mathbf{z}; W, \alpha, L)\,p(L)$.

The prior $p(L)$ encodes beliefs about the plausible range of depths before observing the data.
A convenient and interpretable choice is the geometric prior $p(L) \propto \rho^L$ for some $\rho \in (0, 1)$, which penalizes deeper tree spaces exponentially and favors parsimonious models.
In the absence of strong structural beliefs, a discrete uniform prior over $\{1, \dots, L_{\max}\}$ reduces depth selection to maximizing the marginal likelihood directly. Examples of the use of the method are presented in \ref{appendix-C}.

\section{Simulation Study: Comparing different priors}\label{sec-simulations}

In this section, we present a simulation study comparing different choices of prior distributions in terms of marginal likelihood and MAP estimation.

\subsection{Data and Setup}

We generate binary data from two different VLMC models. The first corresponds to Example 4 in \cite{BerchRaft} (p. 353), and the second is an example of a $0$-renewing tree. Both models are illustrated in Figure \ref{sim-1}-\hyperlink{model1}{(a)} and \ref{sim-1}-\hyperlink{model2}{(b)}, respectively. A similar simulation study is presented in the supplementary material in \cite{kontoyiannis2022bayesiancontexttreesmodelling}, and includes the model \hyperlink{model1}{(a)}, allowing a direct comparison.

Throughout, we refer to the data-generating model as a scenario. In particular, the data generated from model \hyperlink{model1}{(a)} correspond to scenario (a), while the data generated from model \hyperlink{model2}{(b)} correspond to scenario (b). In each scenario, we assume that $\mathbf{z}$ is a sample compatible with the corresponding context tree $(\tau,\boldsymbol{\theta})$ illustrated in Figure \ref{sim-1}. In the simulation setup, we consider different sequence lengths $n = 200, 500, 1000,$ and $2500$, so we analyze two scenarios, each with four sequences, and each sequence has one of the specified lengths. Finally, we set the Dirichlet hyperparameter at $\alpha = 1/2$ and consider the maximal depth $L=10$, resulting in a model space with approximately $1.44 \times 10^{217}$ trees. 

Competing models are defined by selecting different context-tree functions $W$. For simplicity, we denote the $\beta$-BCT function by $B_{\beta}$, the CTW function by $C$, the $\beta$-constant function by $K_\beta$, the $\beta$-exponential function by $E_{\beta}$, the $\beta$-target $l$-depth function by $T_\beta^l$, the $a$-renewal indicator by $I_a$, and the $l$-$m$-depth indicator function by $D_l^m$, see Table \ref{table_ex} for details. The $\beta$ and $l$ parameters utilized in each scenario will be presented in the results.

\begin{figure}[t]
    \centering

    \begin{minipage}{0.3\textwidth}

        \begin{tikzpicture}[
            level distance=1.5cm,
            level 1/.style={sibling distance=7cm},
            level 2/.style={sibling distance=3.5cm},
            level 3/.style={sibling distance=1.75cm},
            grow'=down,
            every node/.style={
                draw,
                circle,
                thick,
                fill=black,
                font=\small,
                align=center,
                text centered,
                minimum size=1pt,
                inner sep=1.5pt
            },
            edge from parent/.style={draw, thick},
            scale=0.60,
            internal/.style={
                draw,
                circle,
                thick,
                fill=white,
                minimum size=1pt,
                inner sep=1.5pt
            },
        ]

        \node[internal] (root) {$\lambda$}
            child { node[internal] {1}
                child { node {\color{white}{1}} }
                child { node[internal] {{0}}
                    child { node {\color{white}{1}} }
                    child { node {\color{white}{0}} }
                }
            }
            child { node[internal] {0}
                child { node[internal] {1}
                    child { node {\color{white}{1}} }
                    child { node {\color{white}{0}} }
                }
                child { node[internal] {0}
                    child { node {\color{white}{1}} }
                    child { node {\color{white}{0}} }
                }
            };

        \begin{scope}[every node/.style={text=black, font=\tiny}]
            \node at (root-2-2-2) [below=8pt] {$(0.9, 0.1)$};
            \node at (root-2-2-1) [below=8pt] {$(0.6, 0.4)$};
            \node at (root-2-1-2) [below=8pt] {$(0.7, 0.3)$};
            \node at (root-2-1-1) [below=8pt] {$(0.3, 0.7)$};
            \node at (root-1-2-2) [below=8pt] {$(0.8, 0.2)$};
            \node at (root-1-2-1) [below=8pt] {$(0.4, 0.6)$};
            \node at (root-1-1) [below=8pt] {$(0.4, 0.6)$};
        \end{scope}

        \begin{scope}[every node/.style={text=black, font=\normalsize}]
            \node at (root) [above=44pt] {\hypertarget{model1}{(a)}};
        \end{scope}

        \end{tikzpicture}
    \end{minipage}
    \hfill
    \hspace{-3cm}
    \begin{minipage}{0.34\textwidth}
        \centering
        \hypertarget{model2}{(b)}

        \begin{tikzpicture}[
            level distance=1.5cm,
            level 1/.style={sibling distance=2cm},
            level 2/.style={sibling distance=2cm},
            level 3/.style={sibling distance=2cm},
            level 4/.style={sibling distance=2cm},
            grow'=down,
            every node/.style={
                draw,
                circle,
                thick,
                fill=black,
                font=\small,
                align=center,
                text centered,
                minimum size=1pt,
                inner sep=1.5pt
            },
            edge from parent/.style={draw, thick},
            scale=0.65,
            internal/.style={
                draw,
                circle,
                thick,
                fill=white,
                minimum size=1pt,
                inner sep=1.5pt
            },
        ]

        \node[internal] (root) {$\lambda$}
            child { node[internal] {1}
                child { node[internal] {{1}}
                    child { node[internal] {{1}}
                        child { node {\color{white}{1}} }
                        child { node {\color{white}{0}} }
                    }
                    child { node {\color{white}{0}} }
                }
                child { node {\color{white}{0}} }
            }
            child { node {\color{white}{0}} };

        \begin{scope}[every node/.style={text=black, font=\tiny}]
            \node at (root-2) [below=8pt] {$(0.1, 0.9)$};
            \node at (root-1-2) [below=8pt] {$(0.5, 0.5)$};
            \node at (root-1-1-2) [below=8pt] {$(0.5, 0.5)$};
            \node at (root-1-1-1-2) [below=8pt] {$(0.5, 0.5)$};
            \node at (root-1-1-1-1) [below=8pt] {$(0.9, 0.1)$};
        \end{scope}

        \end{tikzpicture}
    \end{minipage}

    \caption{VLMC models $(\tau, \boldsymbol{\theta})$ used to generate the data in each simulation scenario.}
    \label{sim-1}
\end{figure}

The CTW and $\beta$-BCT functions are included especially for comparison with the Bayesian framework of \cite{kontoyiannis2022bayesiancontexttreesmodelling}, as the resulting distributions correspond to the prior distributions used in that work. Our primary interest lies in the performance of the $\beta$-target $3$-depth function in scenario~(a) and the $0$-renewal indicator function in scenario~(b), as they induce suitable priors, that is, distributions that assign higher probability to the true trees. The other choices are considered to study different types of penalization, such as constant and length-based penalizations, or excluding trees whose depth is greater than the true depth.

\subsection{Simulation Results Evaluation}

We conduct our analysis using two approaches: for each competing model, we utilize the MAP tree $\tau_{\text{map}}$ as a point estimator and measure the discrepancy between this estimated tree and the generator tree, reporting both prior and posterior probabilities of the true tree $\tau$; and we utilize the base-10 log marginal likelihood as a metric to evaluate the evidence in favor of one competing model over another.

To formally analyze the discrepancy between two trees, we introduce a distance metric based on their structural differences.

\begin{definition}
For any two trees $\tau_1, \tau_2 \in \mathcal{T}_L$, the \textit{structural distance} between $\tau_1$ and $\tau_2$ is defined as $\Delta(\tau_1, \tau_2) := |\tau_1^c \triangle \tau_2^c|$, where $\triangle$ denotes the symmetric difference operator.
\end{definition}

This distance represents the total number of growing or pruning operations required to transform $\tau_1$ into $\tau_2$. To illustrate, consider the binary trees $\tau_1 = \{0, 01, 11\}$ and $\tau_2 = \{00, 10, 1\}$. The corresponding sets of inner nodes are $\tau_1^c = \{\lambda, 1\}$ and $\tau_2^c = \{\lambda, 0\}$. Therefore, the structural distance between $\tau_1$ and $\tau_2$ is $\Delta(\tau_1, \tau_2) = |\{\lambda, 1\} \triangle \{\lambda, 0\}| = |\{0, 1\}| = 2$, which is the number of operations needed to transform $\tau_1$ into $\tau_2$: (i) grow leaf $0$, and (ii) prune inner node $1$.
\cite{ZaninZambom04032025} introduced a similar metric for the purpose of clustering context trees. 

Tables \ref{tab-simula} and \ref{tab-simula-new} present the simulation results, consisting of the distance between the true tree and the MAP tree, the prior and posterior probabilities of the true tree, and the value of $\log_{10}$ of the marginal likelihood for each simulation scenario.  Since $\alpha$ is fixed along the models, we omit it from the notation in the simulation results. 


\begin{table}[ht]
\centering
\caption{Distance between the true tree and the MAP tree $\Delta(\tau, \tau_{\text{map}})$, prior and posterior probabilities of the true tree, $\pi_W(\tau)$ and $\pi_W(\tau \mid \mathbf{z})$ respectively, and $\log_{10} p(\mathbf{z}; W)$ under scenario (a).}
\label{tab-simula}
\vspace{1em}

\scriptsize 

\begin{minipage}{0.49\textwidth}
\centering
\textbf{$n=200$}\par\vspace{0.5em}
\begin{tabular}{c|c|c|c|c}
\toprule
$W$ & $\Delta$ & $\pi_W(\tau)$ & $\pi_W(\tau\mid\mathbf{z})$ & $\log p(\mathbf{z}; W)$ \\
\midrule
$D_0^3$ & 1 & 0.03846 & 0.19471 & -49.60 \\
$C$ & 2 & 1.22e-4 & 0.00602 & -50.59 \\
$B_{0.2}$ & 4 & 3.36e-6 & 3.87e-5 & -49.96 \\
$B_{0.7}$ & 4 & 6.00e-5 & 0.00920 & -51.08 \\
$T_2^3$ & 0 & 0.01738 & 0.11401 & -49.71 \\
$T_3^3$ & 0 & 0.04794 & 0.21042 & -49.54 \\
$T_{8}^3$ & 1 & 0.06817 & 0.22273 & -49.41 \\
$T_2^4$ & 2 & 4.72e-6 & 3.42e-4 & -50.75 \\
$K_{e^{-2}}$ & 4 & 5.15e-6 & 0.00321 & -51.69 \\
$K_{e^{-5}}$ & 6 & 9.29e-14 & 2.13e-9 & -53.26 \\
$E_{-1}$ & 4 & 1.79e-9 & 4.17e-6 & -52.26 \\
\bottomrule
\end{tabular}
\end{minipage}
\hfill
\begin{minipage}{0.49\textwidth}
\centering
\textbf{$n=500$}\par\vspace{0.5em}
\begin{tabular}{c|c|c|c|c}
\toprule
$W$ & $\Delta$ & $\pi_W(\tau)$ & $\pi_W(\tau\mid\mathbf{z})$ & $\log p(\mathbf{z}; W)$ \\
\midrule
$D_0^3$ & 0 & 0.03846 & 0.68559 & -121.15 \\
$C$ & 0 & 1.22e-4 & 0.14506 & -122.98 \\
$B_{0.2}$ & 0 & 3.36e-6 & 0.00392 & -122.97 \\
$B_{0.7}$ & 0 & 6.00e-5 & 0.36007 & -123.68 \\
$T_2^3$ & 0 & 0.01738 & 0.29001 & -121.12 \\
$T_3^3$ & 1 & 0.04794 & 0.34865 & -120.76 \\
$T_{8}^3$ & 1 & 0.06817 & 0.22818 & -120.43 \\
$T_2^4$ & 1 & 4.72e-6 & 0.00704 & -123.08 \\
$K_{e^{-2}}$ & 0 & 5.15e-6 & 0.44670 & -124.84 \\
$K_{e^{-5}}$ & 3 & 9.29e-14 & 0.00238 & -130.31 \\
$E_{-1}$ & 2 & 1.79e-9 & 0.06662 & -127.47 \\
\bottomrule
\end{tabular}
\end{minipage}
\end{table}
\begin{table}[ht]
\centering
\scriptsize
\begin{minipage}{0.49\textwidth}
\centering
\textbf{$n=1000$}\par\vspace{0.5em}
\begin{tabular}{c|c|c|c|c}
\toprule
$W$ & $\Delta$ & $\pi_W(\tau)$ & $\pi_W(\tau\mid\mathbf{z})$ & $\log p(\mathbf{z}; W)$ \\
\midrule
$D_0^3$ & 0 & 0.03846 & 0.72914 & -216.90 \\
$C$ & 0 & 1.22e-4 & 0.34219 & -219.07 \\
$B_{0.2}$ & 0 & 3.36e-6 & 0.14306 & -220.26 \\
$B_{0.7}$ & 0 & 6.00e-5 & 0.52063 & -219.56 \\
$T_2^3$ & 0 & 0.01738 & 0.29979 & -216.86 \\
$T_3^3$ & 1 & 0.04794 & 0.36120 & -216.50 \\
$T_{8}^3$ & 1 & 0.06817 & 0.24175 & -216.18 \\
$T_2^4$ & 1 & 4.72e-6 & 0.00744 & -218.82 \\
$K_{e^{-2}}$ & 0 & 5.15e-6 & 0.70849 & -220.76 \\
$K_{e^{-5}}$ & 0 & 9.29e-14 & 0.95854 & -228.64 \\
$E_{-1}$ & 0 & 1.79e-9 & 0.97203 & -224.36 \\
\bottomrule
\end{tabular}
\end{minipage}
\hfill
\begin{minipage}{0.49\textwidth}
\centering
\textbf{$n=2500$}\par\vspace{0.5em}
\begin{tabular}{c|c|c|c|c}
\toprule
$W$ & $\Delta$ & $\pi_W(\tau)$ & $\pi_W(\tau\mid\mathbf{z})$ & $\log p (\mathbf{z}; W)$ \\
\midrule
$D_3$ & 0 & 0.03846 & 0.92244 & -572.77 \\
$C$ & 0 & 1.22e-4 & 0.68134 & -575.14 \\
$B_{0.2}$ & 0 & 3.36e-6 & 0.29788 & -576.34 \\
$B_{0.7}$ & 0 & 6.00e-5 & 0.81028 & -575.52 \\
$T_2^3$ & 0 & 0.01738 & 0.69354 & -572.99 \\
$T_3^3$ & 0 & 0.04794 & 0.72845 & -572.57 \\
$T_{8}^3$ & 0 & 0.06817 & 0.58963 & -572.33 \\
$T_2^4$ & 0 & 4.72e-6 & 0.18752 & -575.99 \\
$K_{e^{-2}}$ & 0 & 5.15e-6 & 0.88558 & -576.63 \\
$K_{e^{-5}}$ & 0 & 9.29e-14 & 0.99441 & -584.42 \\
$E_{-1}$ & 0 & 1.79e-9 & 0.99345 & -580.14 \\
\bottomrule
\end{tabular}
\end{minipage}
\end{table}



\begin{table}[H]
\centering
\caption{Distance between the true tree and the MAP tree $\Delta(\tau, \tau_{\text{map}})$, prior and posterior probabilities of the true tree ($\pi_W(\tau)$ and $\pi_W(\tau \mid \mathbf{z}))$, and $\log_{10} p(\mathbf{z}; W)$ under scenario (b).}
\label{tab-simula-new}
\vspace{1em}

\scriptsize 

\begin{minipage}{0.49\textwidth}
\centering
\textbf{$n=200$}\par\vspace{0.5em}
\begin{tabular}{c|c|c|c|c}
\toprule
$W$ & $\Delta$ & $\pi_W(\tau)$ & $\pi_W(\tau\mid\mathbf{z})$ & $\log p(\mathbf{z}; W)$ \\
\midrule
$D_0^4$ & 0 & 1.48e-3 & 0.12656 & -52.97 \\
$C$ & 3 & 1.95e-3 & 0.14035 & -52.89 \\
$B_{0.2}$ & 3 & 1.31e-4 & 0.01681 & -53.15 \\
$B_{0.7}$ & 3 & 1.36e-3 & 0.16928 & -53.13 \\
$T^3_2$ & 0 & 1.09e-3 & 0.11153 & -53.05 \\
$I_0$ & 0 & 0.09091 & 0.20753 & -51.40 \\
$T^4_2$ & 0 & 1.89e-5 & 0.00684 & -53.60 \\
$K_{e^{-2}}$ & 3 & 2.81e-4 & 0.06388 & -53.39 \\
$K_{e^{-5}}$ & 4 & 2.05e-9 & 1.52e-6 & -53.91 \\
$E_{-1}$ & 3 & 7.23e-7 & 1.84e-4 & -53.44 \\
\bottomrule
\end{tabular}
\end{minipage}
\hfill
\begin{minipage}{0.49\textwidth}
\centering
\textbf{$n=500$}\par\vspace{0.5em}
\begin{tabular}{c|c|c|c|c}
\toprule
$W$ & $\Delta$ & $\pi_W(\tau)$ & $\pi_W(\tau\mid\mathbf{z})$ & $\log p(\mathbf{z};W)$ \\
\midrule
$D_0^4$ & 1 & 1.48e-3 & 0.04795 & -121.49 \\
$C$ & 1 & 1.95e-3 & 0.07774 & -121.57 \\
$B_{0.2}$ & 0 & 1.31e-4 & 0.01769 & -122.11 \\
$B_{0.7}$ & 1 & 1.36e-3 & 0.18299 & -122.10 \\
$T^3_2$ & 1 & 1.09e-3 & 0.09436 & -121.91 \\
$I_0$ & 0 & 0.09091 & 0.23855 & -120.39 \\
$T^4_2$ & 1 & 1.89e-5 & 0.00397 & -122.30 \\
$K_{e^{-2}}$ & 0 & 2.81e-4 & 0.34525 & -123.06 \\
$K_{e^{-5}}$ & 3 & 2.05e-9 & 7.76e-4 & -125.55 \\
$E_{-1}$ & 3 & 7.23e-7 & 0.01516 & -124.30 \\
\bottomrule
\end{tabular}
\end{minipage}

\vspace{1em}

\begin{minipage}{0.49\textwidth}
\centering
\textbf{$n=1000$}\par\vspace{0.5em}
\begin{tabular}{c|c|c|c|c}
\toprule
$W$ & $\Delta$ & $\pi_W(\tau)$ & $\pi_W(\tau\mid\mathbf{z})$ & $\log p(\mathbf{z}; W)$ \\
\midrule
$D_0^4$ & 0 & 1.48e-3 & 0.42729 & -237.29 \\
$C$ & 0 & 1.95e-3 & 0.43251 & -237.18 \\
$B_{0.2}$ & 0 & 1.31e-4 & 0.22225 & -238.06 \\
$B_{0.7}$ & 0 & 1.36e-3 & 0.62562 & -237.49 \\
$T^3_2$ & 0 & 1.09e-3 & 0.40149 & -237.40 \\
$I_0$ & 0 & 0.09091 & 0.28661 & -235.33 \\
$T^4_2$ & 0 & 1.89e-5 & 0.18344 & -238.82 \\
$K_{e^{-2}}$ & 0 & 2.81e-4 & 0.78394 & -238.27 \\
$K_{e^{-5}}$ & 3 & 2.05e-9 & 0.06964 & -242.36 \\
$E_{-1}$ & 0 & 7.23e-7 & 0.57120 & -240.73 \\
\bottomrule
\end{tabular}
\end{minipage}
\hfill
\begin{minipage}{0.49\textwidth}
\centering
\textbf{$n=2500$}\par\vspace{0.5em}
\begin{tabular}{c|c|c|c|c}
\toprule
$W$ & $\Delta$ & $\pi_W(\tau)$ & $\pi_W(\tau\mid\mathbf{z})$ & $\log p(\mathbf{z}; W)$ \\
\midrule
$D_0^4$ & 0 & 1.48e-3 & 0.19285 & -606.79 \\
$C$ & 0 & 1.95e-3 & 0.63302 & -607.18 \\
$B_{0.2}$ & 0 & 1.31e-4 & 0.59094 & -608.33 \\
$B_{0.7}$ & 0 & 1.36e-3 & 0.75434 & -607.42 \\
$T^3_2$ & 0 & 1.09e-3 & 0.34583 & -607.18 \\
$I_0$ & 0 & 0.09091 & 0.49179 & -605.41 \\
$T^4_2$ & 6 & 1.89e-5 & 0.03533 & -607.94 \\
$K_{e^{-2}}$ & 0 & 2.81e-4 & 0.86566 & -608.16 \\
$K_{e^{-5}}$ & 0 & 2.05e-9 & 0.99403 & -613.36 \\
$E_{-1}$ & 0 & 7.23e-7 & 0.99193 & -610.81 \\
\bottomrule
\end{tabular}
\end{minipage}
\end{table}

\subsection{Results and Comments}  

In simulation scenario (a), at the smallest sample size ($n=200$), only two models correctly identify the true tree via the MAP estimator. In particular, both are specified with $\beta$-target $3$-depth context-tree functions. For larger sample sizes, most of the models identify the true tree, differing only in their posterior probability values. According to the log-evidence metric, model $T_8^3$ performs best in all cases, with small differences among the other $\beta$-target $3$-depth models. This behavior is expected, as the resulting prior distributions assign higher mass to trees close to the maximal tree of depth $3$, varying only through the concentration parameter $\beta$. An important point is that, under the best prior $T_8^3$, the MAP estimate is, in some cases, not the true tree but rather the maximal tree of depth $3$. Nevertheless, the true tree still attains high posterior probability, indicating that the posterior distribution is strongly concentrated on trees that explain the data well. This suggests that higher log-evidence values normally imply good point estimators.

In simulation scenario (b), more models identify the true tree via the MAP estimator, even for small sample sizes. Interestingly, the $\beta$-target $3$-depth models are successful, whereas the $\beta$-BCT models are not. We emphasize model $T_2^4$ at the larger sample size ($n=2500$): observe that the MAP estimate is far from the true tree (differing by $6$ node operations). In fact, the estimated tree is the maximal tree of depth $4$. The prior probability assigned to trees close to it may lead the model to identify dependencies in the data that do not actually exist, suggesting that a misspecified prior can result in worse point estimators even in large samples. In terms of log-evidence, the $0$-indicator model $I_0$ is the best performer and is significantly superior to the others. This is expected, since the resulting distribution assigns prior mass only to $0$-renewing trees. Unlike in the first scenario, under model $I_0$, the MAP estimate coincides with the true tree for all sample sizes.

Overall, for both simulation scenarios, the $l$-$m$-depth indicator models ($D_0^3$ and $D_0^4$) perform well across all sample sizes, outperforming the $\beta$-BCT models both in recovering the true tree and with respect to log-evidence. This illustrates that selecting a uniform prior at the correct depth can be more effective than arbitrarily choosing a prior distribution, and it also highlights the importance of correctly specifying the maximal depth $L$.

The $\beta$-constant and $\beta$-exponential models perform worse for smaller sample sizes under both criteria. For larger samples, however, the true tree attains the highest posterior probability among all models. This is due to the strong penalization induced by these models. Indeed, the greater the penalization, the more posterior mass is assigned to the true tree in larger samples; in small samples, however, heavy penalization makes it difficult to identify the true tree. Therefore, in large-sample applications where structural estimation is desired, choosing these types of models may be advantageous.

The $\beta$-BCT models consistently identify the true tree for larger samples, but they underperform relative to the other models in terms of log-evidence. Moreover, changing the $\beta$ parameter yields very small differences between the $\beta$-BCT models under both criteria.

Finally, as $n$ increases, all models tend to correctly identify the true tree, indicating that with sufficient data, the influence of the prior becomes secondary. Even models such as $T_2^4$ in the first scenario and $T_2^3$ in the second scenario, which correspond to misspecified priors that mistake the target depth by one unit, successfully recover the true tree for larger samples.

\section{Application to Financial Market Data}\label{sec-financial}

To illustrate the proposed Bayesian framework in a real-world setting, we apply it to daily observations of the S\&P 500 index. Financial market data provide a natural setting for investigating variable-length dependence, as the temporal structure of market movements may vary across different market conditions. Rather than seeking to develop a predictive model of financial returns, our objective is to examine whether the proposed framework can identify and compare interpretable context-tree structures in discretized financial sequences. In particular, we consider two representations of the data, one describing the direction of daily market movements and another capturing the intensity of market fluctuations. This setting allows us to assess how alternative prior specifications influence the inferred memory structure and whether the resulting evidence supports a particular depth of temporal dependence.

\subsection{Data Description}

The data consist of daily closing prices of the S\&P 500 index from January 1990 to December 2025, obtained from \cite{YahooFinanceSP500}. Let $P_t$ denote the closing price on day $t$. We compute the daily log-return as
$$
    r_t = \log\left(\frac{P_t}{P_{t-1}}\right).
$$
From the return series, we construct two discrete sequences designed to capture different aspects of market behavior.

TThe first sequence represents the direction and magnitude of daily market movements. Returns are then classified into three categories according to their magnitude relative to a fixed threshold of $0.5\%$,
$
    \mathcal{A}_R = \{\text{D},\text{N},\text{U}\},
$
where D denotes days with returns below $-0.5\%$, U denotes days with returns above $0.5\%$, and N represents values between these two thresholds. This discretization provides a simple and economically interpretable representation of downward, neutral, and upward market movements. The threshold of $0.5\%$ was chosen to distinguish relatively small daily fluctuations from more substantial market movements while preserving a sufficiently large number of observations in each category.

The second sequence focuses on market turbulence. We use the absolute return,
$
    v_t = |r_t|,
$
as a simple proxy for the magnitude of daily fluctuations and classify its values into three categories according to the empirical terciles of the resulting distribution:
$\mathcal{A}_V = \{\text{L},\text{M},\text{H}\}$,
where L, M, and H represent low, medium and high volatility, respectively. Although absolute returns are not a direct measure of conditional volatility, they provide a simple observable proxy for daily market turbulence. The use of terciles produces three categories with approximately equal marginal frequencies and avoids imposing an arbitrary scale on the volatility measure. This representation is also motivated by the well-documented persistence of volatility in financial time series, commonly referred to as volatility clustering.

The resulting sequences are analyzed separately using the proposed framework.
For each sequence, the original numerical observations are therefore replaced by a finite alphabet representation, allowing the context tree model to characterize the dependence structure of market movements and fluctuations.

\subsection{Priors and Setup}

For our analysis, we consider $\alpha = 1/2$ and set the maximal depth to $L=10$, capturing memory dependencies of up to ten days (two trading weeks).
For both datasets, we consider the same set of competing priors based on context-tree functions: the unity (uniform) function $U$, the CTW $C$, the $\beta$-BCT functions $B_{0.25}$, $B_{0.75}$ and $B_{0.90}$, the $(-1/5)$-exponential $E_{-1/5}$ and the $3$-target $3$-depth $T_3^3$.

The inclusion of $T_3^3$ is motivated by the hypothesis that relevant dependence in financial market data may be concentrated over a short temporal horizon. More specifically, this prior provides a direct mechanism for assigning greater prior weight to context-tree structures associated with depth three. This makes $T_3^3$ useful for investigating whether a relatively short and specific memory horizon is supported by the data. In contrast, the $\beta$-BCT family controls local branching probabilities through the parameter $\beta$, rather than directly targeting a particular context depth. The two prior families therefore encode different assumptions about the structure of temporal dependence.

The remaining prior specifications serve complementary purposes. The CTW and $\beta$-BCT functions are included to provide a direct comparison with the established framework of \cite{kontoyiannis2022bayesiancontexttreesmodelling}.
The unity function $U$ serves as a baseline specification with no preferential weighting among the admissible structures. Finally, the exponential function $E_{-1/5}$ represents an alternative prior that penalizes longer contexts through an explicit dependence on context length.

All prior specifications are evaluated using the same data, alphabet, value of $\alpha$, and maximal depth $L$, so that differences in marginal likelihood can be attributed to the prior distribution assigned to the context-tree structures. The resulting base-10 log marginal likelihoods are then compared across prior specifications for each of the two discrete sequences.

\subsection{Results and Comments}

The log marginal likelihood values are reported in Table \ref{tab:daily}. For both datasets, the prior induced by $T_3^3$ achieves the highest marginal likelihood among the specifications considered.
For the daily return sequence, $T_3^3$ yields a log marginal likelihood of $-4052.370$, followed by $E_{-1/5}$ with $-4053.167$.
For the volatility sequence, the corresponding values are $-4273.204$ and $-4276.119$, respectively. Thus, the ranking of the two leading specifications is consistent across the two representations of the market data.

The $\beta$-BCT specifications perform consistently below $T_3^3$ and $E_{-1/5}$. Among the three values considered, $B_{0.75}$ achieves the highest marginal likelihood for both sequences, whereas $B_{0.25}$ and $B_{0.90}$ perform less favorably. This indicates that the marginal likelihood does not vary monotonically with $\beta$ over the values considered. More generally, the results suggest that controlling local branching probabilities through a single $\beta$ parameter does not provide the same degree of support for the observed memory structure as the more targeted $T_3^3$ specification.

The CTW prior also performs below $T_3^3$ and $E_{-1/5}$, although its performance is relatively close to that of some of the $\beta$-BCT specifications. In contrast, the unity prior $U$ yields by far the lowest marginal likelihood in both datasets. This indicates that, among the prior specifications considered, the data are better accommodated by priors that place some preference on the space of context trees rather than treating the competing structures uniformly.

\begin{table}[H]
\centering
 \scriptsize
\caption{Log marginal likelihood values under each prior distribution for daily return ($\mathbf{z}_R$) and daily volatility ($\mathbf{z}_V$) data.}
 \vspace{0.5em}
\begin{tabular}{c|ccccccc}
 \toprule
 $W$ & $C$ & $B_{0.25}$ & $B_{0.75}$ & $B_{0.90}$ & $U$ & $E_{-1/5}$ & $T_3^3$ \\
 \midrule
 $\log_{10} p(\mathbf{z}_R; W)$ & -4057.239 & -4057.323 & -4054.874 & -4056.484 & -4237.191 & -4053.167 & -4052.370 \\
 \midrule
    $\log_{10} p(\mathbf{z}_V; W)$ & -4281.084 & -4280.501 & -4278.374 & -4280.238 & -4351.729 & -4276.119 & -4273.204  \\
 \bottomrule
 \end{tabular}
 \label{tab:daily}
\end{table}
 
The results therefore provide evidence in favor of a context-tree prior that concentrates on a short and specific memory horizon. Importantly, this conclusion does not depend solely on the performance of $T_3^3$.
As a further assessment of the preferred memory depth, we apply the maximal-depth selection procedure from Subsection \ref{subsec:max_depth} to both datasets, assigning a discrete uniform prior to $L\in\{1,2,\ldots,10\}$ and a uniform prior to $\tau$. In both cases, the procedure selects $L^\ast=3$.
The agreement between the prior comparison and the maximal-depth selection procedure provides additional evidence that depth three is a particularly well-supported memory scale for both the return and volatility sequences.

\section{Discussion}\label{sec-conclusion}

The main contribution of this paper is to show that prior distributions over context trees can be specified directly in terms of structural node weights, rather than indirectly through branching probabilities. This is conceptually attractive because prior beliefs in VLMCs are naturally expressed in terms of properties of the tree itself, such as its depth, sparsity, or preference for specific contexts. In contrast, branching-process priors induce these effects only implicitly through local splitting rules, which makes prior elicitation harder to interpret and calibrate. The node-weight representation provides a more transparent and flexible mechanism for encoding structural assumptions while preserving the recursive computations needed for exact Bayesian inference.

A second key advantage of the proposed framework is its unifying character. The CTW and BCT priors, as well as more general branching-process constructions, can be recovered as special cases of the broader class of context-tree functions. At the same time, the framework also allows priors that are difficult to represent in the branching formulation, such as those based on depth preferences. This flexibility is important because the prior is not only a regularization device, but also encodes substantive assumptions about the underlying dependence structure. The ability to express these assumptions directly in terms of tree features is therefore central to the practical usefulness of the methodology.

The computational results show that, by exploiting the factorization of the prior and the Dirichlet-integrated likelihood, the proposed methodology yields exact recursive algorithms for computing marginal likelihoods, posterior probabilities, and MAP trees without enumerating the full tree space. This enables exact Bayes factors for comparing alternative structural priors, which is a major methodological advantage over simulation-based procedures. In this sense, the node-weight formulation is not only more interpretable, but also statistically more informative, because it supports principled model comparison and helps distinguish between competing structural hypotheses in a way that is difficult to achieve under more restrictive prior constructions.

The proposed framework has a minor limitation: the representation is not unique. Different weight functions may induce the same distribution over trees, and therefore the numerical values of the weights are not themselves uniquely interpretable. This is not a defect of the method, but rather a consequence of the fact that only the relative scores of trees matter for the induced prior. As a result, prior specification should be guided by structural meaning and application-specific considerations, rather than by the raw magnitude of the weights alone. In addition, the choice of the weight function remains a modeling decision and must be assessed carefully.

The availability of an open-source implementation makes the proposed methodology readily applicable in practice and facilitates the development of new structural priors within the proposed framework.
The proposed methods are implemented in the \texttt{bacontrees} R package, available at \url{https://github.com/Freguglia/bacontrees} and on CRAN. Reproducibility scripts for all numerical results are available \href{https://github.com/PaulichenT/vlmc_priors-repro-arxiv}{here}.

\begin{funding}
This work was financed in part by the Coordenação de Aperfeiçoamento de Pessoal de Nível Superior - Brasil (CAPES) - Finance Code 001, by the Conselho Nacional de Desenvolvimento Científico e Tecnológico (CNPq), under grant number 141012/2024-2 and by FAEPEX (Fundo de Apoio ao Ensino, à Pesquisa e à Extensão), Universidade Estadual de Campinas (Unicamp), under grant number 3387/23.
\end{funding}


\bibliographystyle{ba}
\bibliography{sample}

\appendix
\section*{Appendices}

\renewcommand\thefigure{\thesection.\arabic{figure}}  
\renewcommand\thetable{\thesection.\arabic{table}}  

\subsection{Proofs of the results} \label{appendix-A}

In this appendix, we present the proofs of the results of the work. Throughout, we fix a maximal depth $L$. 

\subsubsection{Proof of Lemma 1}

\begin{proof}[Proof of Lemma 1]
Under the assumptions of the lemma we have that:
\begin{align*}
    p_\alpha(\mathbf{z}\mid\tau) &= \int_{\Delta_m^{|\tau|}} p(\mathbf{z} \mid \tau, \boldsymbol{\theta}) \pi_\alpha(\boldsymbol{\theta} \mid \tau) d\boldsymbol{\theta} \\
    &= \int_{\Delta_m^{|\tau|}}  \prod_{\mathbf{s} \in \tau} \frac{\Gamma(m \alpha)}{\Gamma(\alpha)^m} \prod_{k=0}^{m-1} \theta_{\mathbf{s}}(k)^{c_{\mathbf{s}k}(\mathbf{z}) + \alpha - 1} d\boldsymbol{\theta} \\
    &= \prod_{\mathbf{s} \in \tau} 
    \int_{\Delta_m} \frac{\Gamma(m \alpha)}{\Gamma(\alpha)^m} \prod_{k=0}^{m-1} \theta_{\mathbf{s}}(k)^{c_{\mathbf{s}k}(\mathbf{z}) + \alpha - 1}  d\boldsymbol{\theta}_{\mathbf{s}} \\
    &= \prod_{\mathbf{s} \in \tau} \frac{\Gamma(m \alpha)}{\Gamma(\alpha)^m} \int_{\Delta_m} \prod_{k=0}^{m-1} \theta_{\mathbf{s}}(k)^{(c_{\mathbf{s}k}(\mathbf{z}) + \alpha) - 1} d\boldsymbol{\theta}_{\mathbf{s}} \\
    &= \prod_{\mathbf{s} \in \tau} \frac{\Gamma(m \alpha)}{\Gamma(\alpha)^m} \frac{\prod_{k=0}^{m-1} \Gamma\big(c_{\mathbf{s}k}(\mathbf{z}) + \alpha\big)}{\Gamma \left( \sum_{k=0}^{m-1} \big(c_{\mathbf{s}k}(\mathbf{z}) + \alpha\big) \right)} \\
    &= \prod_{\mathbf{s} \in \tau} 
    \frac{\Gamma(m \alpha)}{\Gamma(\alpha)^m} 
    \frac{\prod_{k=0}^{m-1} \Gamma \big(c_{\mathbf{s}k}(\mathbf{z}) + \alpha\big)}{\Gamma \left(\sum_{k=0}^{m-1} c_{\mathbf{s}k}(\mathbf{z}) + m\alpha \right)}.
\end{align*}
The fourth equality follows by recognizing that, for each $\mathbf{s} \in \tau$, the integrand is proportional to a Dirichlet density with parameter vector $(c_{\mathbf{s}0}(\mathbf{z}) + \alpha, c_{\mathbf{s}1}(\mathbf{z}) + \alpha, \dots, c_{\mathbf{s}(m-1)}(\mathbf{z}) + \alpha)$ and then multiplying by the corresponding normalizing constants. 
\end{proof}

\subsubsection{Proof of Theorem 1} \label{A3}

From now on, we will denote by $\mathcal{T}_L^{\mathbf{r}}$ the set of all trees with root node $\mathbf{r}$ and depth less than or equal to $L$.

\begin{proof}[Proof of Theorem 1]
    First, for a fixed $L$, define for each node $\mathbf{r} \in \mathcal{N}(\tau_{\text{max}})$:
    \begin{align} \label{eq-app}
        \Sigma_W(\mathbf{r}) = \sum_{\tau \in \mathcal{T}^{\mathbf{r}}_{L-\ell(\mathbf{r})}} \prod_{\mathbf{s} \in \tau} w(\mathbf{s}).
    \end{align}
    In particular, we can rewrite the set $\mathcal{T}^{\mathbf{r}}_{L-\ell(\mathbf{r})}$ as $\{\mathbf{r}\} \cup \mathcal{T}^{\mathbf{r}}_{L-\ell(\mathbf{r})} \setminus \{\mathbf{r}\}$, where $\{\mathbf{r}\}$ represents the tree that consists only of the root node $\mathbf{r}$. Moreover, for any tree $\tau \in \mathcal{T}_{L-\ell(\mathbf{r})}^{\mathbf{r}} \setminus \{{\mathbf{r}}\}$, for each $k \in \{0, \dots, m-1\}$ we can define the set:
    \begin{align*}
        \tau^k = \{\mathbf{s}k: \mathbf{s}k \in \tau\},
    \end{align*}
    which is nonempty by the full property, defines a tree in $\mathcal{T}^{k\mathbf{r}}_{L-\ell({\mathbf{r}})-1}$  and from the definition $\tau = \bigcup_{k=0}^{m-1} \tau^k$. Indeed, the function that assigns $\tau$ to the $m$-tuple $(\tau^0, \dots, \tau^{m-1})$ is a bijection between $\mathcal{T}_{L-\ell(\mathbf{r})}^{\mathbf{r}} \setminus \{{\mathbf{r}}\}$ and $\mathcal{T}^{0{\mathbf{r}}}_{L-\ell({\mathbf{r}})-1} \times \cdots \times \mathcal{T}^{(m-1){\mathbf{r}}}_{L-\ell({\mathbf{r}})-1}$. See an example of the bijection map in Figure \ref{fig1}.

    \begin{figure}[h!]
    \centering
    \begin{tikzpicture}[
        level distance=2cm,
        level 1/.style={sibling distance=5cm},
        level 2/.style={sibling distance=2.5cm},
        level 3/.style={sibling distance=1.25cm},
        grow'=down,
        every node/.style={
            draw,
            circle,
            thick,
            fill=white,
            font=\small,
            align=center,
            text centered,
            minimum size=1pt,
            inner sep=1.5pt
        },
        edge from parent/.style={draw, thick},
        scale=0.6,
        internal/.style={
            draw,
            circle,
            thick,
            fill=white,
            minimum size=1pt,
            inner sep=1.5pt
        },
    ]
    \node[style={draw, thick, opacity=0.3}] (root)  {{{$\lambda$}}}
        child[edge from parent/.style={draw, thick, opacity=0.3}]  { node[style={draw, thick, opacity=0.3}] {{1}}
            child { node {{1}} 
                child { node {{1}} }
                child { node {{0}} }
                }
            child { node {{0}}
                child { node {{1}} }
                child { node {{0}} }}
        }
        child[edge from parent/.style={draw, thick, opacity=0.3}] { node {{0}}
            child[edge from parent/.style={draw, thick, opacity=1}] { node[style={draw, thick, opacity=1, fill = black}] {\color{white}{1}} 
                child[edge from parent/.style={draw, thick, opacity=0.3}] { node[style={draw, thick, opacity=0.3}] {{1}} }
                child[edge from parent/.style={draw, thick, opacity=0.3}] { node[style={draw, thick, opacity=0.3}] {{0}} }}
            child[edge from parent/.style={draw, thick, opacity=1}] { node[style={draw, thick, opacity=1}] {{0}} 
                child { node[style={fill = black}] {\color{white}{1}} }
                child { node[style={fill = black}] {\color{white}{0}} }
            }
        };

    \begin{scope}[every node/.style={text=black, font=\small}]
        \node at (root) [below, yshift = -4cm] {$\tau = \{00\underline{0}, 10\underline{0}, 1\underline{0}\}$};
        \node at (root) [below, yshift = -4cm, xshift = 7cm] {$(\tau^0, \tau^1) = (\{0\underline{00}, 1\underline{00}\}, \{\underline{10}\})$}; 
    \end{scope}

    \draw[very thick,->] ([xshift=20pt]root-1-1.east) -- ++(4.5,0);

    \begin{scope}[xshift=11cm, yshift=-3cm, 
        level distance=2cm,
        level 1/.style={sibling distance=1.5cm},
        grow'=down,
        every node/.style={
            draw,
            circle,
            thick,
            fill=white,
            font=\small,
            align=center,
            text centered,
            minimum size=1pt,
            inner sep=1.5pt
        },
        edge from parent/.style={draw, thick},
        scale=0.6,
        internal/.style={
            draw,
            circle,
            thick,
            fill=white,
            minimum size=1pt,
            inner sep=1.5pt
        },
    ]
    \node (t1) {00}
        child { node[style={fill = black}] {\color{white}{1}} }
        child { node[style={fill = black}] {\color{white}{0}} };
    \end{scope}

    \begin{scope}[xshift=13cm, yshift=-3cm,
        level distance=2cm,
        level 1/.style={sibling distance=1.25cm},
        grow'=down,
        every node/.style={
            draw,
            circle,
            thick,
            fill=black,
            font=\small,
            align=center,
            text centered,
            minimum size=1pt,
            inner sep=1.5pt
        },
        edge from parent/.style={draw, thick},
        scale=0.6,
        internal/.style={
            draw,
            circle,
            thick,
            fill=white,
            minimum size=1pt,
            inner sep=1.5pt
        },
    ]
    \node (t2) {\color{white}{10}};
    \end{scope}
    
    \end{tikzpicture}
    
    \caption{Example of the bijection map for $L=3$ and $\mathcal{A}=\{0, 1\}$. The tree $\tau$ is an element of the set $\mathcal{T}^{0}_{2}$ while $\tau^0 \in \mathcal{T}^{00}_{1}$ and $\tau^1 \in \mathcal{T}^{10}_{1}$. The pair $(\tau^0, \tau^1)$ is the correspondent element to $\tau$ in $\mathcal{T}^{00}_{1} \times \mathcal{T}^{10}_{1}$. Here, the underlined symbols are representing the roots.}  
    \label{fig1}
\end{figure}
    
    Therefore, for a node $\mathbf{r}$ such that $\ell(\mathbf{r}) < L$, we can rewrite:
    \begin{align*}
         \Sigma_W(\mathbf{r}) &= \sum_{\tau \in \{\mathbf r\} \cup \mathcal{T}_{L-\ell(\mathbf{r})}^{\mathbf r} \setminus \{\mathbf r\}} \prod_{\mathbf{s} \in \tau} w(\mathbf{s}) \\
                              &= w(\mathbf{r}) + \sum_{\tau \in \mathcal{T}_{L-\ell(\mathbf{r})}^{\mathbf{r}} \setminus \{\mathbf{r}\}} \prod_{\mathbf{s} \in \tau} w(\mathbf{s}) \\
                              &= w(\mathbf{r}) + \sum_{(\tau^{0}, \dots, \tau^{m-1}) \in \mathcal{T}_{L-\ell(\mathbf{r})-1}^{0\mathbf{r}} \times \cdots \times \mathcal{T}_{L-\ell(\mathbf{r})-1}^{(m-1)\mathbf{r}}} \left( \prod_{\mathbf{s} \in \cup_{k=0}^{m-1} \tau^k} w(\mathbf{s}) \right) \\
                              &= w(\mathbf{r}) + \sum_{(\tau^0, \dots, \tau^{m-1}) \in \mathcal{T}_{L-\ell(\mathbf{r})-1}^{0\mathbf{r}} \times \cdots \times \mathcal{T}_{L-\ell(\mathbf{r})-1}^{(m-1)\mathbf{r}}} \left( \prod_{\mathbf{s} \in \tau^0} w(\mathbf{s}) \cdots \prod_{\mathbf{s} \in \tau^{m-1}} w(\mathbf{s}) \right) \\
                              &= w(\mathbf{r}) + \sum_{\tau^0 \in \mathcal{T}_{L-\ell(\mathbf{r})-1}^{0\mathbf{r}}} \cdots \sum_{\tau^{m-1} \in \mathcal{T}_{L-\ell(\mathbf{r})-1}^{(m-1)\mathbf{r}}} \prod_{k=0}^{m-1} \prod_{\mathbf{s} \in \tau^i} w(\mathbf{s}) \\
                              &= w(\mathbf{r}) + \prod_{k=0}^{m-1} \sum_{\tau \in \mathcal{T}^{k\mathbf{r}}_{L-\ell(\mathbf{r})-1}} \prod_{\mathbf{s} \in \tau} w(\mathbf{s}) \\
                              &= w(\mathbf{r}) + \prod_{k=0}^{m-1} \Sigma_W(k\mathbf{r}),
    \end{align*}
    and for a leaf node $\mathbf{r}$, that is, $\ell(\mathbf{r}) = L$, we have:
    \begin{align*}
        \Sigma_W(\mathbf{r}) = \sum_{\tau \in \mathcal{T}^{\mathbf{r}}_{0}} \prod_{\mathbf{s} \in \tau} w(\mathbf{s}) = w(\mathbf{r})
    \end{align*}
    since $\mathcal{T}^{\mathbf{r}}_{0} = \{\{\mathbf{r}\}\}$. Thus, the result follows by evaluating the function $\Sigma_W$ at the root $\lambda$:
    \begin{align*}
        \Sigma_W(\lambda) =  \sum_{\tau \in \mathcal{T}_L^\lambda} W(\tau) .
    \end{align*}
\end{proof}

As pointed out in Subsection \ref{subsec:recur}, the proof for the maximum follows similarly. In particular, it suffices to define, for each node $\mathbf{r} \in \mathcal{N}(\tau_{\text{max}})$,
\begin{align*}
    \Upsilon_W(\mathbf{r}) = \max_{\tau \in \mathcal{T}_{L - \ell(\mathbf{r})}^{\boldsymbol{r}}} \prod_{\mathbf{s} \in \tau} w(\mathbf{s}).
\end{align*}
Then, by replacing the summation with the maximum in the recursive equations, the same arguments guarantee that
\begin{align*}
    \Upsilon_W(\lambda) = \max_{\tau \in \mathcal{T}_L^\lambda} W(\tau). 
\end{align*}

\subsubsection{Proof of Proposition 1}

\begin{proof}[Proof of Proposition 1]
First, note that the resulting tree obtained after inspecting a specific node $\mathbf{r}\in\mathcal{N}(\tau_{\mathrm{max}})$ in the step: \lq\lq \textit{Prune the sub-tree below $\mathbf{r}$ if
\begin{align*}
        w(\mathbf{r}) \ge \prod_{k=0}^{m-1}\Upsilon_W(k\mathbf{r});
\end{align*}
Otherwise continue the same test for each child $k\mathbf{r}$, $k=0,\dots,m-1$}\rq\rq \ can be defined via the following recursive function:

\begin{align*}
T(\mathbf{r}) \in \mathcal{T}^{\mathbf{r}}_{L-\ell(\mathbf{r})}    
\end{align*}
where
\begin{itemize}
    \item if $\ell(\mathbf{r})=L$:
    \begin{align*}
        T(\mathbf{r})=\{\mathbf{r}\};
    \end{align*}
    \item otherwise, if $\ell(\mathbf{r})<L$:
    \begin{align*}
    T(\mathbf{r})=
    \begin{cases}
        \{\mathbf{r}\}, 
        & \text{if } w(\mathbf{r}) \geq \prod_{k=0}^{m-1} \Upsilon_W(k\mathbf{r}),\\[4pt]
        \displaystyle\bigcup_{k=0}^{m-1}T(k\mathbf{r}),
        & \text{otherwise}.
    \end{cases}
\end{align*}
\end{itemize}
This function is well defined by the bijection map in Proof of Theorem 1. Indeed, it ensures that, given a collection of trees $(T(0\mathbf{r}), \dots, T((m-1)\mathbf{r})) \in \mathcal{T}^{0\mathbf{r}}_{L-\ell(\mathbf{r})-1} \times \cdots \times  \mathcal{T}^{(m-1)\mathbf{r}}_{L-\ell(\mathbf{r})-1}$ there exists a tree $ \tau \in \mathcal{T}^{\mathbf{r}}_{L-\ell(\mathbf{r})}$, such that $\bigcup_{k=0}^{m-1} T(k\mathbf{r}) = \tau$. An illustration of the recursive procedure involving function $T$ is presented in Figure \ref{fig2}

\begin{figure}
    \centering
    \begin{tikzpicture}[
        level distance=2cm,
        level 1/.style={sibling distance=2.5cm},
        level 2/.style={sibling distance=1.25cm},
        grow'=down,
        scale=0.6,
        edge from parent/.style={draw, thick},
        every node/.style={
            draw,
            circle,
            thick,
            fill=black,
            font=\small,
            align=center,
            text centered,
            minimum size=1pt,
            inner sep=1.5pt
        },
        internal/.style={
            draw,
            circle,
            thick,
            fill=white,
            minimum size=1pt,
            inner sep=1.5pt
        },
        labelnode/.style={draw=none, fill=none}
    ]

    \node[internal] (root) {$\lambda$}
        child { node[internal] {1}
            child { node {\color{white}{1}} }
            child { node {\color{white}{0}} }
        }
        child { node[internal, dashed] {0}
            child { node {\color{white}{1}} }
            child { node {\color{white}{0}} }
        };

    \draw[very thick, dashed, ->] ([xshift=25pt]root-1.east) -- ++(7.5,2.5)
        node[midway, above, labelnode, yshift = -0.7cm] {$ w(0)\geq \Upsilon_W(00)\cdot \Upsilon_W(10) $};

    \draw[very thick, dashed, ->] ([xshift=25pt]root-1.east) -- ++(7.5,-2.5)
        node[midway, below, labelnode, yshift = 0.7cm] {$ w(0)< \Upsilon_W(00)\cdot \Upsilon_W(10) $};

    \begin{scope}[xshift=12cm, yshift=0.5cm]
        \node (t1) {\color{white}{$0$}};
    \end{scope}

    \begin{scope}[xshift=11.5cm, yshift=-4.5cm]
        \node (t2) {\color{white}{$00$}};
    \end{scope}

    \begin{scope}[xshift=13cm, yshift=-4.5cm]
        \node (t3) {\color{white}{$10$}};
    \end{scope}

    \begin{scope}[every node/.style={text=black, font=\small}]
        \node at (root) [below, yshift=-3cm]
            {$\tau_{\text{max}} = \{00, 10, 01, 11\}$};

        \node at (t1) [below, yshift=-0.5cm]
            {$\{\underline{0}\}$};

        \node at (t2) [below, yshift=-0.5cm, xshift=0.5cm]
            {$(\{\underline{00}\}, \{\underline{10}\})$};
    \end{scope}

    \end{tikzpicture}
    \vspace{-1.5cm}
    \caption{Illustration of the procedure involving the function $T$ for $L=2$ and $\mathcal{A} = \{0,1\}$. The node considered is $\mathbf{r}=0$ (shown as dashed in the figure). Note that if $w(0)\geq \Upsilon_W(00)\cdot \Upsilon_W(10)$, then $T(0)=\{\underline{0}\} \in \mathcal{T}_1^0$. Otherwise, the bijection map ensures, $T(0)=\{\underline{00}\}\cup \{\underline{10}\} \simeq \{0\underline{0},1\underline{0}\} \in \mathcal{T}_1^0$. The underlined symbols are representing the roots.}
    \label{fig2}
\end{figure}

Now, we can prove, by descending induction on $\ell(\mathbf{r})$, that for every $\mathbf{r} \in \mathcal{N}(\tau_{\text{}})$,
\begin{align}\label{ind}
    \Upsilon_W(\mathbf{r}) = \max_{\tau \in \mathcal{T}^{\mathbf{r}}_{L-\ell(\mathbf{r})}} \prod_{\mathbf{s} \in \tau} w(\mathbf{s}) = W(T(\mathbf{r})) = \prod_{\mathbf{s}\in T(\mathbf{r})} w(\mathbf{s}),
\end{align}
and therefore, the result will be obtained by evaluating the function $T$ at the root $\lambda$. Indeed, if $\ell(\mathbf{r})=L$, then $\mathcal{T}^{\mathbf{r}}_{0}=\{\{\mathbf{r}\}\}$, and hence
\begin{align*}
    \Upsilon_W(\mathbf{r}) = w(\mathbf{r}) = W(\{\mathbf{r}\}) = W(T(\mathbf{r})),  
\end{align*}
so \eqref{ind} holds. Then, assume that \eqref{ind} is valid for all nodes of length $0 < l \leq L$. Let $\mathbf{r}$ be a node that satisfies $\ell(\mathbf{r}) = l - 1$. By the definition of $\Upsilon_W$,
\begin{align*}
    \Upsilon_W(\mathbf{r}) = \max\left\{w(\mathbf{r}), \prod_{k=0}^{m-1} \Upsilon_W(k\mathbf{r})\right\}.
\end{align*}
Therefore, we have two cases:
\begin{enumerate}
    \item[(i)] If $w(\mathbf{r}) \geq \prod_{k=0}^{m-1}\Upsilon_W(k\mathbf{r})$, then $T(\mathbf{r})=\{\mathbf{r}\}$ and so:
    \begin{align*}
        W(T(\mathbf{r}))=w(\mathbf{r})= \max\left\{ w(\mathbf{r}), \prod_{k=0}^{m-1} \Upsilon_W(k\mathbf{r})\right\} = \Upsilon_W(\mathbf{r}).
    \end{align*}
    \item[(ii)] Otherwise, $w(\mathbf{r}) < \prod_{k=0}^{m-1}\Upsilon_W(k\mathbf{r})$, and by the definition of $T$:
    \begin{align*}
        T(\mathbf{r})=\bigcup_{k=0}^{m-1}T(k\mathbf{r}).
    \end{align*}
    Using the induction hypothesis on each child $k\mathbf{r}$:
    \begin{align*}
        W(T(\mathbf{r})) &= \prod_{\mathbf{s} \in \bigcup_{k=0}^{m-1}T(k\mathbf{r})} w(\mathbf{s}) \\ 
                         &= \prod_{k=0}^{m-1} \prod_{\mathbf{s} \in T(k\mathbf{r})} w(\mathbf{s}) \\
                         &= \prod_{k=0}^{m-1} \Upsilon_W(k\mathbf{r}) \\
                         &= \max\left\{ w(\mathbf{r}), \prod_{k=0}^{m-1} \Upsilon_W(k\mathbf{r})\right\} \\ 
                         &=\Upsilon_W(\mathbf{r}).
    \end{align*}
\end{enumerate}
In all cases, $W(T(\mathbf{r})) = \Upsilon_W(\mathbf{r})$, proving \eqref{ind}. Thus, the result follows by defining $\tau^\ast = T(\lambda)$. 
\end{proof}

\subsection{Comparing exact posterior distribution with Metropolis-Hastings approximation} \label{appendix-b}

In this appendix, our goal is to quantify the statistical costs obtained by adopting a Bayesian approach that avoids the computation of the marginal likelihood, beyond the obvious loss of standard Bayesian model selection. For this purpose, we evaluate the performance of the Metropolis--Hastings algorithm proposed by \cite{kontoyiannis2022bayesiancontexttreesmodelling}, on a simulated dataset.

\subsubsection{Data and setup}

We consider samples of size $n = 500, 1000, 1500, 2000, 2500, 3000, 3500, 4000, 4500, 5000$ from the VLMC model (a) presented in Figure \ref{sim-1}. The Bayesian models are specified considering as fixed parameters the maximal depth $L = 6$ and the Dirichlet hyperparameter $\alpha = 1/2$. In this case, the model space $\mathcal{T}_L$ has approximately $2.10 \times 10^{11}$ trees. We specify models by four different context-tree functions, which are the unity $U$, the $2$-target $3$-depth $T_2^3$, the CTW $C$, and the $(-1/2)$-exponential $E_{-1/2}$, see Table \ref{table_ex}. We run the Metropolis--Hastings algorithm under the same setup of \cite{VictorNancy} for $10{,}000$ iterations with a burn-in period of $1{,}000$ iterations, yielding a sequence of $M = 9{,}000$ sampled trees, denoted by $\tau^{(1)}, \tau^{(2)}, \dots, \tau^{(M)}$. Finally, we denote the set of distinct sampled trees by $\mathcal{T}_{\mathrm{mh}}=\left\{\tau^{(1)}, \tau^{(2)}, \dots, \tau^{(M)}\right\}$.

\subsubsection{Simulation results evaluation}

To evaluate the performance of the Metropolis--Hastings algorithm, we consider three quantities. The first is the number of distinct trees sampled by the algorithm, $N_{\mathrm{mh}}=|\mathcal{T}_{\mathrm{mh}}|$, which offers an indication of the diversity of the subset of trees explored. The second quantity is the \textit{total posterior probability mass} recovered by the Metropolis--Hastings algorithm,
\begin{align*}
P_{\mathrm{mh}}=\sum_{\tau \in \mathcal{T}_{\mathrm{mh}}} \pi_{W, \alpha}(\tau \mid \mathbf{z}),
\end{align*}
which measures the proportion of the posterior distribution represented by the sampled trees. Finally, we define the \textit{normalized MAP error} as
\begin{align*}
E_{\mathrm{MAP}}=
\frac{\left|\pi_{W, \alpha}(\tau_{\mathrm{map}}\mid\mathbf{z})- \frac{1}{M} \sum_{i=1}^M\mathbf{1}\{\tau^{(i)} = \tau_{\text{map}} \}\right|}{P_{\mathrm{mh}}}.
\end{align*}
This quantity measures the discrepancy between the exact posterior probability of the MAP tree and its empirical posterior probability estimated from the Metropolis--Hastings sample, normalized by the total posterior probability mass recovered by the algorithm. The simulation results are presented in Table \ref{tab:n_trees}, Figure \ref{fig:placeholder} and Table \ref{tab:normalized_map_error}. 

\begin{figure}[h]
    \centering
    \includegraphics[scale=0.6]{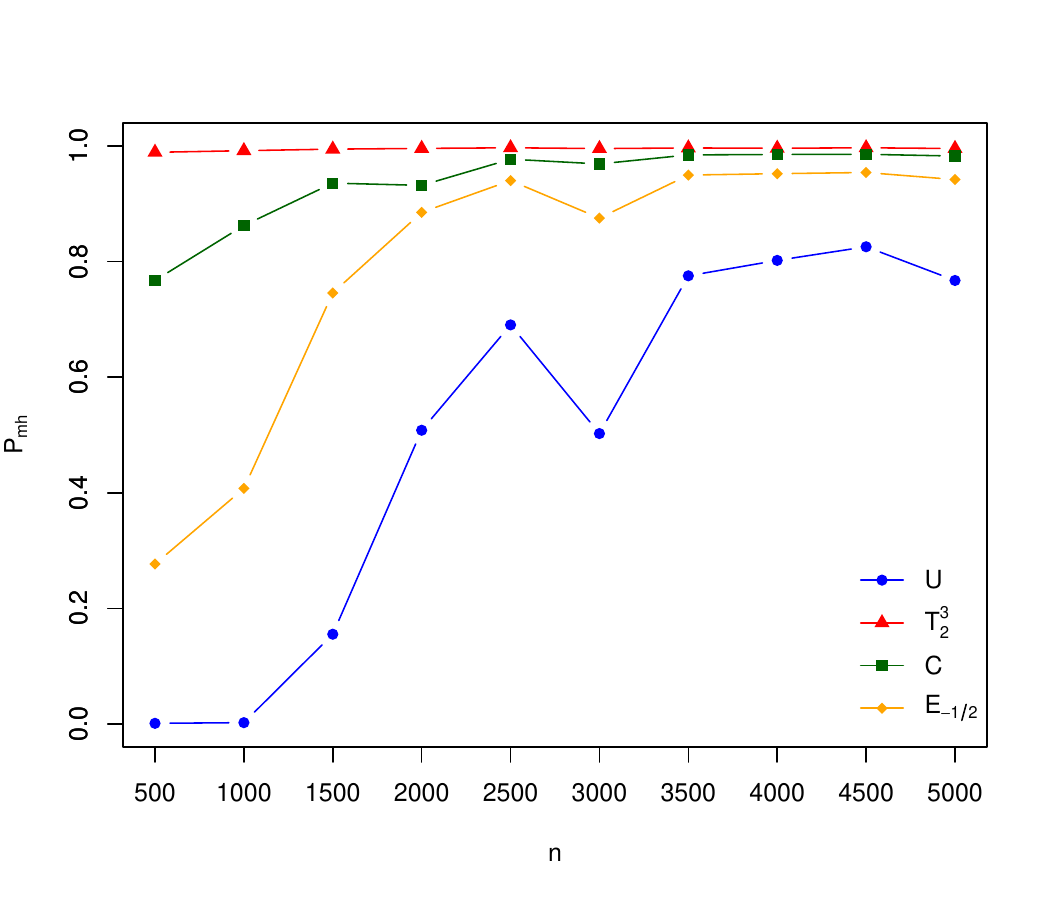}
    \vspace{-0.5cm}
    \caption{Total posterior probability mass recovered by the Metropolis--Hastings algorithm for different sample sizes and prior distributions.}
    \label{fig:placeholder}
\end{figure}

\begin{table}[h]
\centering
\caption{Number of distinct trees sampled by the Metropolis-Hastings algorithm for different sample sizes and prior distributions.}
\vspace{0.2cm}
\scriptsize
\label{tab:n_trees}
\begin{tabular}{c|cccccccccc}
\multicolumn{1}{c}{} & \multicolumn{10}{c}{$n$} \\
\cmidrule(lr){2-11}
$W$ & 500 & 1000 & 1500 & 2000 & 2500 & 3000 & 3500 & 4000 & 4500 & 5000 \\
\midrule
$U$       & 6090 & 5962 & 4749 & 2715 & 1859 & 2694 & 1461 & 1409 & 1320 & 1273 \\ \hline
$T_2^3$   &   66 &   68 &   49 &   41 &   27 &   38 &   40 &   33 &   31 &   25 \\ \hline
$C$       & 1186 &  687 &  371 &  155 &  106 &  146 &   84 &   84 &   78 &   89 \\ \hline
$E_{-1/2}$ & 4135 & 3710 & 1732 &  855 &  457 &  732 &  343 &  364 &  346 &  389 
\end{tabular}
\end{table}

\begin{table}[h]
\centering
\caption{Normalized MAP error for different sample sizes and prior distributions.}
\vspace{0.2cm}
\scriptsize
\label{tab:normalized_map_error}
\begin{tabular}{c|cccccccccc}
\multicolumn{1}{c}{} & \multicolumn{10}{c}{$n$} \\
\cmidrule(lr){2-11}
$W$ & 500 & 1000 & 1500 & 2000 & 2500 & 3000 & 3500 & 4000 & 4500 & 5000 \\
\midrule
$U$       & 3.1795 & 0.8175 & 0.0019 & 0.0155 & 0.0017 & 0.0007 & 0.0113 & 0.0012 & 0.0579 & 0.0129 \\ \hline
$T_2^3$   & 0.0072 & 0.0279 & 0.0131 & 0.0131 & 0.0198 & 0.0044 & 0.0111 & 0.0030 & 0.0034 & 0.0459 \\ \hline
$C$       & 0.0200 & 0.0186 & 0.0203 & 0.0197 & 0.0193 & 0.0096 & 0.0111 & 0.0311 & 0.0156 & 0.0318 \\ \hline
$E_{-1/2}$ & 0.0325 & 0.0226 & 0.0147 & 0.0317 & 0.0163 & 0.0161 & 0.0097 & 0.0287 & 0.0033 & 0.0230
\end{tabular}
\end{table}

\subsubsection{Results and comments}

The number $N_{\text{mh}}$ in Table \ref{tab:n_trees}, demonstrates that increasing the sample size leads the algorithm to visit fewer distinct trees. This occurs naturally as the likelihood function contributes more relative to the prior. A high value of $N_{\text{mh}}$ indicates that the algorithm wanders across the search space, transitioning from tree to tree in the absence of evidence. The prior choice is closely related with this space exploration. Non-informative priors ($U$ and $E_{-1/2}$) induce diffuse posterior distributions that drive broader exploration and higher values of $N_{\text{mh}}$. In contrast, informative and suitable priors ($C$ and $T_2^3$) concentrate the posterior mass, restricting the algorithm to a significantly smaller set of candidate structures.

The use of non-informative priors presents additional drawbacks regarding posterior coverage, see Figure \ref{fig:placeholder}. The values of $P_{\text{mh}}$ indicate that under the uniform prior, the algorithm captures at most $80\%$ of the total posterior mass, even at large sample sizes. This outcome has direct implications: because a non-negligible portion of the posterior probability remains dispersed across unvisited or under-sampled candidate structures, the algorithm can either underestimate or overestimate the probabilities of individual trees. 

On the other hand, under suitable priors, such as $T_2^3$, the algorithm recovers nearly all of the posterior mass. Showing the practical utility of using the Metropolis--Hastings method for posterior sampling in these regimes, since it is computationally cheaper than obtaining the exact quantities. The values of $E_{\text{MAP}}$ in Table \ref{tab:normalized_map_error} also demonstrate that the algorithm is good in estimating the MAP tree probability. Even when $P_{\text{mh}}$ is low, the estimation of the MAP tree probability remains good, producing a substantial error only under the uniform prior $U$ at small sample sizes. 

However, from a Bayesian perspective, point estimation is not the main focus, and this can easily be done under our Bayesian framework via the MAP estimator algorithm. Moreover, our results show that the algorithm exhibits important limitations under non-informative priors, including incomplete posterior coverage and degraded sampling efficiency. Therefore, while MCMC can be an efficient tool for point estimation and approximate posterior sampling under suitable priors, it cannot replace exact Bayesian approaches.

\subsection{Selecting maximal depths} \label{appendix-C}

In this appendix, to exemplify the use of the maximal depth selection procedure from Subsection \ref{subsec:max_depth}, we apply the method to the simulated data from Section \ref{sec-simulations}. As setup, we consider for $L$ a discrete uniform prior over $\{1, 2, \dots, 10\}$ and for $\tau$ a uniform prior over $\mathcal{T}_L$. This reduces the depth selection to maximizing the marginal likelihood over the unity context-tree function. Table \ref{tab:depths} summarizes the maximal depths selected in each simulation scenario.

\begin{table}[h]
    \centering
    \small
    \caption{Selected maximal depths in each simulation scenario and for each sample size.}
    \vspace{0.5em}
    \begin{tabular}{c|c|c|c|c} 
     & $n=200$ & $n=500$ & $n=1000$ & $n=2500$ \\ 
    \midrule
    Scenario & \multicolumn{4}{c}{Selected maximal depth $L$} \\
    \toprule
    (a) & 3 & 3 & 3 & 3 \\ \hline
    (b) & 4 & 4 & 4 & 4 \\ 
    \bottomrule
    \end{tabular}
    \label{tab:depths}
\end{table}

Although selecting the true tree depth is not the primary focus of the method, the results demonstrate its consistency, as it correctly identifies the true depth for both scenarios, $\ell(\tau) = 3$ for scenario (a) and $\ell(\tau) = 4$ for scenario (b).

\end{document}